\documentclass[sigconf]{aamas}

\usepackage{balance} %
\usepackage{tikz,pgffor}
\usetikzlibrary{arrows,backgrounds,calc}
\usetikzlibrary{automata,positioning,arrows,shapes,math,arrows.meta,decorations.pathmorphing}
\tikzstyle{min}=[thick,circle,draw,minimum size=1.3em,inner sep=0em,text centered]

\makeatletter
\def\@copyrightpermission{}
\def\@copyrightowner{}
\def\@copyrightyear{}
\def\@mkbibcitation{}
\def\maketitle{\@beginmaketitlehook
  \@ACM@maketitle@typesettrue
  \if@ACM@anonymous
    \ifnum\num@authorgroups=0\author{}\fi
  \fi
  \begingroup
  \let\@vspace\@vspace@orig
  \let\@vspacer\@vspacer@orig
  \let\@footnotemark\@footnotemark@nolink
  \let\@footnotetext\@footnotetext@nolink
  \renewcommand\thefootnote{\@fnsymbol\c@footnote}%
  \hsize=\textwidth
  \def\@makefnmark{\hbox{\@textsuperscript{\@thefnmark}}}%
  \@mktitle\if@ACM@sigchiamode\else\@mkauthors\fi\@mkteasers
  \@printtopmatter
  \if@ACM@sigchiamode\@mkauthors\fi
  \setcounter{footnote}{0}%
  \def\@makefnmark{\hbox{\@textsuperscript{\normalfont\@thefnmark}}}%
  \@titlenotes
  \@subtitlenotes
  \@authornotes
  \let\@makefnmark\relax
  \let\@thefnmark\relax
  \let\@makefntext\noindent
  \ifx\@empty\thankses\else
    \footnotetextauthorsaddresses{%
      \def\par{\let\par\@par}\parindent\z@\@setthanks}%
  \fi
  \ifx\@empty\@authorsaddresses\else
     \if@ACM@anonymous\else
       \if@ACM@journal@bibstrip
         \footnotetextauthorsaddresses{%
           \def\par{\let\par\@par}\parindent\z@\@setauthorsaddresses}%
       \fi
     \fi
  \fi
  \if@ACM@nonacm\else\footnotetextcopyrightpermission{%
    \if@ACM@authordraft
        \raisebox{-2ex}[\z@][\z@]{\makebox[0pt][l]{\large\bfseries
            Unpublished working draft. Not for distribution.}}%
       \color[gray]{0.9}%
    \fi
    \parindent\z@\parskip0.1\baselineskip
    \if@ACM@authorversion\else
      \if@printpermission\@copyrightpermission\par\fi
    \fi
    \if@ACM@manuscript\else
       \if@ACM@journal@bibstrip\else %
       {\itshape \acmConference@name}. %
       \fi
    \fi
     \fi}
    \fi
  \endgroup
  \setcounter{footnote}{0}%
  \@mkabstract
  \if@ACM@printccs
  \ifx\@concepts\@empty\else\bgroup
      {\@specialsection{CCS Concepts}%
         \noindent\@concepts\par}\egroup
     \fi
   \fi
   \ifx\@keywords\@empty\else\bgroup
      {\if@ACM@journal
         \@specialsection{Additional Key Words and Phrases}%
       \else
         \@specialsection{Keywords}%
       \fi
         \noindent\@keywords\par}\egroup
   \fi
  \let\metadata@authors=\authors
  \nxandlist{, }{, }{, }\metadata@authors
  \def\@ACM@checkaffil{}%
  \hypersetup{%
    pdfauthor={\metadata@authors},
    pdftitle={\@title},
    pdfsubject={\@concepts},
    pdfkeywords={\@keywords},
    pdfcreator={LaTeX with aamas
      \csname ver@aamas.cls\endcsname\space
      and hyperref
      \csname ver@hyperref.sty\endcsname}}%
  \andify\authors
  \andify\shortauthors
  \global\let\authors=\authors
  \global\let\shortauthors=\shortauthors
  \if@ACM@printacmref
     \@mkbibcitation
  \fi
  \global\@topnum\z@ %
  \global\@botnum\z@ %
  \@printendtopmatter
  \@afterindentfalse
  \@afterheading
}
\makeatother

\acmConference{A full version of the paper published in Proc. of AAMAS 2022}
\acmYear{}
\acmDOI{}
\acmPrice{}
\acmISBN{}

\acmSubmissionID{493}

\title[Patrolling]{Minimizing Expected Intrusion Detection Time\\ in Adversarial Patrolling}

\author{David Kla\v{s}ka}
\affiliation{
  \institution{Masaryk University}
  \city{Brno}
  \country{Czechia}}
\email{david.klaska@mail.muni.cz}

\author{Anton\'{\i}n Ku\v{c}era}
\affiliation{
  \institution{Masaryk University}
  \city{Brno}
  \country{Czechia}}
\email{tony@fi.muni.cz}

\author{V\'{\i}t Musil}
\affiliation{
  \institution{Masaryk University}
  \city{Brno}
  \country{Czechia}}
\email{musil@fi.muni.cz}

\author{Vojt\v{e}ch \v{R}eh\'{a}k}
\affiliation{
  \institution{Masaryk University}
  \city{Brno}
  \country{Czechia}}
\email{rehak@fi.muni.cz}

\begin{abstract}
In adversarial patrolling games, a mobile Defender strives to discover intrusions at vulnerable targets initiated by an Attacker.
The Attacker's utility is traditionally defined as the probability of completing an attack, possibly weighted by target costs.
However, in many real-world scenarios, the actual damage caused by the Attacker depends on the \emph{time} elapsed since the attack's initiation to its detection.
We introduce a formal model for such scenarios, and we show that the Defender always has an \emph{optimal} strategy achieving maximal protection.
We also prove that \emph{finite-memory} Defender's strategies are sufficient for achieving protection arbitrarily close to the optimum.
Then, we design an efficient \emph{strategy synthesis} algorithm based on differentiable programming and gradient descent.
We evaluate the efficiency of our method experimentally.
\end{abstract}


\newcommand{\tm}{\mathit{tm}}
\newcommand{\Nset}{\mathbb{N}}

\newcommand{\Rset}{\mathbb{R}}
\newcommand{\Exp}{\mathbb{E}}

\newcommand{\calB}{\mathcal{B}}
\newcommand{\calE}{\mathcal{E}}
\newcommand{\calF}{\mathcal{F}}
\newcommand{\calI}{\mathcal{I}}

\newcommand{\calL}{\mathcal{L}}

\newcommand{\calH}{\mathcal{H}}
\newcommand{\calW}{\mathcal{W}}
\newcommand{\calD}{\mathcal{D}}

\newcommand{\wait}{\mathit{wait}}
\newcommand{\attack}{\mathit{attack}}

\newcommand{\Val}{\operatorname{Val}}

\newcommand{\Obs}{\Omega}

\newcommand{\walk}{\mathit{walk}}

\newcommand{\NP}{\mathsf{NP}}

\newcommand{\PSPACE}{\mathsf{PSPACE}}

\newcommand{\prob}{\mathbb{P}}

\newcommand{\Dist}{\mathit{Dist}}
\newcommand{\Attack}{\mathit{Att}}

\newcommand{\dhat}[1]{\smash{\hat{#1}}}

\newcommand{\Reg}{\mathit{Reg}}

\newcommand{\newmin}{\mathop{\mathrm{min}\vphantom{\mathrm{sup}}}}

\newcommand\mem{\operatorname{mem}}
\let\hat\widehat

\usepackage{xspace}
\makeatletter
\DeclareRobustCommand\onedot{\futurelet\@let@token\@onedot}
\def\@onedot{\ifx\@let@token.\else.\null\fi\xspace}

\newcommand\ie{{i.e}\onedot}

\makeatother

\newcommand{\tran}[1]{\stackrel{\raisebox{-.3ex}{\scriptsize$#1$}}{\rightarrow}}

\graphicspath{{graphics/}}

\usepackage[
	textsize=tiny,
	textwidth=18mm,
	colorinlistoftodos,
	backgroundcolor=teal!30!white,
	linecolor=teal!50!white
	]{todonotes}

\usepackage{algorithm}
\usepackage[noend]{algpseudocode}
\newcommand{\commentsymbol}{//}%
\algrenewcommand\algorithmiccomment[1]{\hfill\textcolor{cyan}{\commentsymbol{} #1}}
\newcommand{\LineComment}[2][1]{\Statex \hspace{#1\dimexpr\algorithmicindent}\textcolor{cyan}{\commentsymbol{} #2}}

\begin{document}

\pagestyle{fancy}
\fancyhead{}

\maketitle 

\section{Introduction}
\label{sec-intro}

This paper follows the \emph{security games} line of work studying optimal allocation of limited security resources for achieving optimal target coverage \cite{Tambe:book}.
Practical applications of security games include the deployment of police checkpoints at the Los Angeles International Airport \cite{PJMOPTWPK:Deployed-ARMOR}, the scheduling of federal air marshals over the U.S.{} domestic airline flights \cite{TRKOT:IRIS}, the arrangement of city guards in Los Angeles Metro \cite{DJYZTKS:patrolling-uncertainty-JAIR}, the positioning of U.S.{} Coast Guard patrols to secure selected locations \cite{ASYTBDMM:Protect-AIMagazine}, and also applications to wildlife protection~\cite{FKDYT:PAWS,WSYWSJF:Patrolling-learning,Xu:Green-security}.

\emph{Patrolling games} are a special type of security games where a mobile Defender moves among protected targets with the aim of detecting possible incidents. Compared with static monitoring facilities such as sensor networks or surveillance systems, patrolling is more flexible and less costly on implementation and maintenance. Due to these advantages \cite{YJC:multi-robot-coordination-survey}, patrolling is indispensable in detecting crimes \cite{JVP:maritime-security-IS,CCW:cooperative-police-patrol}, managing disasters \cite{MCCMO:multiUAV-disaster},  wildlife protection \cite{WSYWSJF:Patrolling-learning,Xu:Green-security}, etc. Many works consider human Defenders such as police squads or rangers \cite{WSYWSJF:Patrolling-learning} where the patrolling horizon is bounded. Recent technological advances motivate the study of robotic patrolling with unbounded horizon where the Defender is an autonomous device operating for a long time without interruption.  

Most of the existing patrolling models can be classified as either \emph{regular} or \emph{adversarial}.
Regular patrolling can be seen as a form of surveillance where the Defender aims at discovering accidents as quickly as possible by minimizing the time lag between two consecutive visits for each target. Here, a Defender's strategy is typically a single path or a cycle visiting all targets. 
In adversarial patrolling, the Defender strives to protect the targets against an Attacker exploiting the best attack opportunities maximizing the damage. The solution concept is typically based on Stackelberg equilibrium \cite{YKKCT:Stackelberg-Nash-security,SFAKT:Stackelberg-Security-Games}, where the Defender commits to a strategy~$\gamma$ and the Attacker follows by selecting a strategy $\pi$ maximizing the expected Attacker's utility against~$\gamma$. Defender's strategies are typically \emph{randomized} so that the Attacker cannot foresee the next Defender's moves, and the Defender aims at maximizing the probability of discovering an attack before its completion. The adversarial model is also appropriate in situations when a certain protection level must be guaranteed even if the accidents happen at the least convenient moment.

In infinite-horizon adversarial patrolling models, every target~$\tau$ is assigned a finite \emph{resilience} $d(\tau)$, and an attack at $\tau$ is discovered if the Defender visits $\tau$ in the next $d(\tau)$ time units. Although this model is adequate in many scenarios, it is not applicable when the actual damage depends on the \emph{time elapsed since initiating the attack}. For example, if the attack involves setting a fire, punching a hole in a fuel tank, or setting a trap, then the associated damage \emph{increases with time}. In this case, the Defender should aim at \emph{minimizing the expected attack discovery time} rather than maximizing the probability of visiting a target before a deadline. We refer to Section~\ref{sec-evaluating} for a more detailed discussion.

In this work, we formalize the objective of minimizing the expected attack discovery time in infinite-horizon adversarial patrolling, and we design an efficient strategy synthesis algorithm. 
We start by fixing a suitable formal model. The terrain is modeled by the standard patrolling graph, and the Defender's/Attacker's strategies are also defined in the standard way. However, the expected damage caused by attacking a target $\tau$ is defined as the expected time of visiting $\tau$ by the Defender since initiating the attack, multiplied by the target cost $\alpha(\tau)$. Intuitively, $\alpha(\tau)$ is the ``damage per time unit'' when attacking $\tau$. We use Stackelberg equilibrium as the underlying solution concept, and define the \emph{protection value} of a given Defender's strategy $\gamma$ as the  expected attack discovery time (weighted by target costs) guaranteed by $\gamma$ against an arbitrary Attacker's strategy.  
 
In general, a Defender's strategy $\gamma$ may randomize and the choice of the next move may depend on the whole history of moves. The randomization is crucial for increasing protection value (a concrete example is given in Section~\ref{sec-evaluating}). Since general strategies are not finitely representable, they are not algorithmically workable. Recent results on infinite-horizon adversarial patrolling \cite{KKLR:patrol-gradient,KKMR:Regstar-UAI} identify the subclass of \emph{regular} Defender's strategies as sufficiently powerful to maximize the probability of timely attack discovery. Here, a strategy is \emph{regular} if it uses finite memory and rational probability distributions. However, it is not clear whether regular strategies are equivalently powerful as general strategies when minimizing the expected attack discovery time. Perhaps surprisingly, we show that the answer is \emph{positive}, despite all issues caused by specific properties of this objective. More precisely, we prove that \emph{the limit protection value achievable by regular strategies is the same as the limit protection value achievable by general strategies}. This non-trivial result is based on deep insights into the structure of (sub)optimal Defender's strategies.

Our second main contribution is an algorithm synthesizing a regular Defender's strategy and its protection value for a given patrolling graph. We show that the protection value of a regular strategy is a differentiable function, and we proceed by designing an efficient strategy improvement procedure based on gradient descent. 

We evaluate our algorithm experimentally on instances of considerable size. Since our work initiates the study of infinite-horizon adversarial patrolling with the expected attack discovery time, there is no baseline set by previous works. To estimate the quality of the constructed regular strategies, we consider instances where the optimal protection value can be determined by hand, but constructing the associated Defender's strategy is sufficiently tricky to properly examine the capabilities of our strategy synthesis algorithm. 

The experiments also show that our algorithm is sufficiently fast for recomputing a patrolling strategy \emph{dynamically} when the underlying patrolling graph changes due to unpredictable external events. Hence, the applicability of our results is not limited just to static scenarios. 
\smallskip

\noindent
\textbf{Our main contribution} can be summarized as follows:
\begin{itemize}
  \item We propose a formal model for infinite-horizon adversarial patrolling where the damage caused by attacking a target depends on the time needed to discover the attack.
  \item We prove that regular strategies can achieve the same limit protection value as general strategies.
  \item We design an efficient algorithm synthesizing a regular Defender's strategy for a given patrolling graph, and we evaluate its functionality experimentally.   
\end{itemize}

\subsection{Related work}
The literature on regular and adversarial patrolling is rich; the existing overviews include
\cite{HZHH:multi-robot-patrol-survey,ARSTMCC:multi-patrolling-survey,PR:multi-patrolling-survey}. We give a summary of previous results on infinite-horizon adversarial patrolling, which is perhaps closest to our work. 

Most of the existing results concentrate on computing an optimal moving strategy for certain topologies of admissible moves. The underlying solution concept is the \emph{Stackelberg equilibrium} \cite{SFAKT:Stackelberg-Security-Games,YKKCT:Stackelberg-Nash-security}, where the Defender/Attacker play the roles of the \mbox{Leader/Follower}. 

For general topologies, the existence of a perfect Defender's strategy discovering all attacks in time is $\PSPACE$-complete \cite{HO:UAV-problem-PSPACE}. Consequently, computing an optimal Defender's strategy is $\PSPACE$-hard. Moreover, computing an $\varepsilon$-optimal strategy for $\varepsilon \leq 1/2n$, where $n$ is the number of vertices, is $\NP$-hard \cite{KKR:patrol-drones}. Hence, no feasible strategy synthesis algorithm can \emph{guarantee} (sub)optimality for all inputs, and finding high-quality strategy in reasonable time is challenging.

The existing strategy synthesis algorithms are based either on mathematical programming, reinforcement learning, or gradient descent. The first approach suffers from scalability issues caused by non-linear constraints \cite{BGA:large-patrol-AI,BGA:patrolling-arbitrary-topologies}. Reinforcement learning has so far been successful mainly for patrolling with finite horizon, such as green security games \cite{WSYWSJF:Patrolling-learning,BAVT:learn-preventive-healthcare,Xu:Green-security,KMZGA:moving_targets}. 
Gradient descent techniques for finite-memory strategies \cite{KL:patrol-regular,KKLR:patrol-gradient,KKMR:Regstar-UAI} are applicable to patrolling graphs of reasonable size. 

Patrolling for restricted topologies has been studied for lines, circles  \citep{AKK:multi-robot-perimeter-adversarial,ASKK:perimeter-patrol}, or fully connected environments \citep{BKR:patrol-Internet}. Apart of special topologies, specific variants and aspects of the patrolling problem have been studied, including moving targets \cite{bosansky2011aamas,Fang2013}, multiple patrolling units \cite{Basilico2010}, movement of the Attacker on the graph \cite{Basilico2009-2}, or reaction to alarms \cite{MunozdeCote2013,BNG:patrolling-alarm}.

\section{The Model}
\label{sec-model}

In this section we introduce a formal model for infinite-horizon patrolling where the Defender aims at minimizing the expected attack discovery time. The terrain (protected area) is modeled by the standard patrolling graph \cite{BGA:patrolling-arbitrary-topologies,BGA:large-patrol-AI,KKLR:patrol-gradient,KKMR:Regstar-UAI}. We use the variant where time is spent by performing edges \cite{KKLR:patrol-gradient,KKMR:Regstar-UAI} rather than by staying in vertices \cite{BGA:patrolling-arbitrary-topologies,BGA:large-patrol-AI}. The model of Defender/Attacker is also standard \cite{KL:patrol-regular,KKR:patrol-drones} (the study of patrolling models usually starts by considering the scenario with one Defender and one Attacker, and we stick to this approach). The new ingredient of our model in the way of evaluating the protection achieved by Defender's strategies (see Section~\ref{sec-evaluating}), and here we devote more space to explaining and justifying our definitions.

In the rest of this paper, we use $\Nset$ and $\Nset_+$ to denote the sets of non-negative and positive integers. We assume familiarity with basic notions of probability, Markov chain theory, and calculus.

\subsection{Terrain model}

Locations where the Defender selects the next move are modeled as vertices in a directed graph. The edges correspond to admissible moves, and are labeled by the corresponding traversal time. The protected targets are a subset of vertices with integer weights representing their costs. Formally, a \emph{patrolling graph} is a tuple $G = (V,T,E,\tm,\alpha)$ where
\begin{itemize}
	\item $V$ is a finite set of \emph{vertices} (Defender's positions);
	\item $T \subseteq V$ is a non-empty set of \emph{targets}; 
	\item $E \subseteq V \times V$ is a set of \emph{edges} (admissible moves);
	\item $\tm\colon E\to\Nset_+$ specifies the traversal time of an edge;
	\item $\alpha\colon T\to\Rset_+$ defines the costs of targets.
\end{itemize}
We require that $G$ is strongly connected, i.e., for all $v,u \in V$ there is a path from $v$ to $u$. We write \mbox{$u\to v$} instead of $(u,v)\in E$, and use $\alpha_{\max}$ to denote the maximal target cost.

The set of all non-empty finite and infinite paths in $G$ are denoted by $\calH$ (\emph{histories}) and $\calW$ (\emph{walks}), respectively. For a given history $h = v_1,\ldots,v_n$, we use $\tm(h) = \sum_{i=1}^{n-1} \tm(v_i, v_{i+1})$ to denote the total traversal time of~$h$.

\subsection{Defender and Attacker}
\label{sec-Defender-strategy}
We adopt a simplified patrolling scenario with one Defender and one Attacker. In the rest of this section, let $G$ be a fixed patrolling graph.

\subsubsection{Defender.}
\label{sec-defender}
A \emph{Defender's strategy} is a function $\gamma$ assigning to every history $h \in \calH$ of Defender's moves a probability distribution on $V$ such that $\gamma(h)(v) > 0$ only if $hv \in \calH$, i.e., $u \to v$ where $u$ is the last vertex of~$h$. We also use $\walk(h)$ to denote the set of all walks initiated by a given $h \in \calH$.

For every \emph{initial vertex} $v$ where the Defender starts patrolling, the strategy $\gamma$ determines the probability space $(\calW,\calF,\prob^{\gamma,v})$ over the walks in the standard way, i.e., $\calF$ is the $\sigma$-field generated by all $\walk(h)$ where $h \in \calH$, and $\prob^{\gamma,v}$ is the unique probability measure satisfying $\prob^{\gamma,v}(\walk(h)) = \prod_{i=1}^{n-1}\gamma(v_1,\ldots,v_i)(v_{i+1})$ for every history $h = v_1,\ldots,v_n$ where $v_1 = v$ (if $v_1 \neq v$, we have that $\prob^{\gamma,v}(\walk(h)) = 0$). We use $\Exp^{\gamma,v}[R]$ to denote the expected value of a random variable $R$ in this probability space.

\subsubsection{Attacker.}
The Attacker observes the history of Defender's moves and decides whether and where to initiate an attack. In general, the Attacker may be able to determine the next  Defender's move right after the Defender leaves the vertex currently visited. For the Attacker, this is the best moment to attack, because initiating an attack in the middle of a Defender's move can only decrease the Attacker's utility (cf.{} Section~\ref{sec-evaluating}).

Formally, an \emph{observation} is a sequence $o =
v_1,\ldots, v_n, v_n {\rightarrow} v_{n+1}$, where $v_1,\ldots, v_n$ is a path in~$G$. Intuitively, $v_1,\ldots, v_n$ is the sequence of vertices visited by the Defender, $v_n$ is the currently visited vertex, and $v_n {\rightarrow} v_{n+1}$ is the edge taken next. 
The set of all observations is denoted by $\Obs$. 

An \emph{Attacker's strategy} is a function
\mbox{$\pi\colon \Obs \rightarrow \{\wait,\attack_\tau:\tau\in T\}$}. We require that
if $\pi(v_1,\ldots, v_n,v_n {\rightarrow} u) = \attack_\tau$ for some $\tau \in T$,
then $\pi(v_1,\ldots, v_i,v_i {\rightarrow} v_{i+1}) = \wait$ for all $1\leq
i<n$, i.e., the Attacker exploits  an optimal attack opportunity.

\subsection{Protection value} 
\label{sec-evaluating}

Suppose the Defender commits to a strategy $\gamma$ and the Attacker selects a strategy $\pi$. The \emph{expected damage} caused by $\pi$ against $\gamma$ is the expected time to discover an attack scheduled by $\pi$ weighted by target costs. 

More precisely, we say that a target $\tau$ is \emph{attacked along a walk $w = v_1,v_2,\ldots$} if  $\pi(v_1,\ldots,v_n,v_n\rightarrow v_{n+1}) = \tau$ for some $n$. Note that the index $n$ is unique if it exists. Let $m > n$ be the least index such that $v_m = \tau$. If no such $m$ exists, we say that the attack along $w$ is \emph{not discovered}. Otherwise, the attack is \emph{discovered in time} $\tm(v_{n},\ldots,v_m)$. 

Let $\calD^{\pi} : \calW \rightarrow \Nset_{\infty}$ be a function defined as follows:
\begin{equation*}
\calD^{\pi}(w) =
\begin{cases}
   \ell\cdot\alpha(\tau) & \parbox[t]{.6\textwidth}{if $\tau$ is attacked along $w$ and the attack is\\ discovered in time $\ell$;}\\   
   \infty & \parbox[t]{.6\textwidth}{if $\tau$ is attacked along $w$ and the attack is not\\ discovered;}\\
   0 & \mbox{if no target is attacked along $w$.}
\end{cases}
\end{equation*}
The expected damage caused by $\pi$ against $\gamma$ initiated in $v$ is defined as $\Exp^{\gamma,v}[\calD^{\pi}]$. Since the Defender may choose the initial vertex~$v$, we define the \emph{protection value achieved by $\gamma$} and the \emph{limit protection value} as follows:
\begin{eqnarray}
   \Val(\gamma) & = & \newmin_{v} \sup_{\pi}\    \Exp^{\gamma,v}[\calD^{\pi}]\label{eq:valgamma-def}\\[1ex]
   \Val         & = & \inf_{\gamma}\ \Val(\gamma)
\end{eqnarray} 
We say that a Defender's strategy $\gamma$ is \emph{optimal} if $\Val(\gamma) = \Val$.
 
\subsubsection{Discussion}

In this section, we discuss possible alternative approaches to formalizing the objective of discovering an initiated attack as quickly as possible. 

Note that this objective is implicitly taken into account in regular patrolling where the Defender aims at minimizing the time lag between two consecutive visits for each target (see Section~\ref{sec-intro}). For randomized strategies, one may try to minimize the \emph{expected time lag} between two consecutive visits for each target. At first glance, this objective seems similar to minimizing $\sup_{\pi}\ \Exp^{\gamma,v}[\calD^{\pi}]$. In reality, the objective is different and problematic. To see this, consider the trivial patrolling graph of Fig.~\ref{fig-renewal}a with two targets $\tau_1,\tau_2$ and four edges (incl.{} two self-loops) with traversal time~$1$. The costs of both targets are equal to~$1$. A natural strategy $\gamma_1$ for patrolling these targets is a deterministic loop alternately visiting $\tau_1$ and $\tau_2$  (see Fig~\ref{fig-renewal}b). Then, the maximal expected time lag between two consecutive visits of a target is~$2$, and we also have that $\Val(\gamma_1) = 2$. However, consider the randomized strategy $\gamma_2$ of Fig.~\ref{fig-renewal}c. In the target currently visited, the Defender performs the self-loop with probability $0.99$, and with the remaining probability $0.01$, the Defender moves to the other target (see the dashed arrows in Fig.~\ref{fig-renewal}c). For $\gamma_2$, the maximal expected time lag between two consecutive visits of a target is \emph{again equal to~$2$} (in Markov chain terminology, the stationary distribution determined by $\gamma_2$ assigns $1/2$ to each target, and hence the mean recurrence time is equal to $2$ for each target). Hence, if we adopted minimizing the maximal expected time lag between two consecutive visits of a target as the Defender's objective, the strategies $\gamma_1$ and $\gamma_2$ would be \emph{equivalently good}, despite the fact that the expected time to visit $\tau_2$ from $\tau_1$ is $100$ when the Defender commits to $\gamma_2$. Observe that the difference between $\gamma_1$ and $\gamma_2$ is captured properly by our approach.

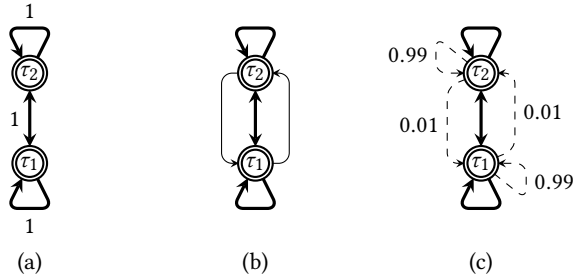
\begin{figure}[t]\centering
\begin{tikzpicture}[x=3cm, y=1.2cm]
\foreach \x/\c/\l in {0/0/a,1/1/b,2/2/c}{%
    \coordinate (a) at (\c,0);
    \node at ($(a) +(0,-1.1)$) {(\l)};
    \node [min,double] (T1\x) at (a) {$\tau_1$};
    \node [min,double] (T2\x) at ($ (a) + (0,1) $) {$\tau_2$};
    \ifthenelse{\x=0}{\draw[stealth-stealth,very thick] (T1\x) -- node[left] {$1$}  (T2\x);
                      \draw [-stealth,very thick,rounded corners] (T1\x) -- +(.1,-.5) --node[below] {$1$} +(-.1,-.5) -- (T1\x);  
                      \draw [-stealth,very thick,rounded corners] (T2\x) -- +(.1,.5) --node[above] {$1$} +(-.1,.5) -- (T2\x); 
                     }%
                     {\draw[stealth-stealth,very thick] (T1\x) --  (T2\x);
                      \draw [-stealth,very thick,rounded corners] (T1\x) -- +(.1,-.5) -- +(-.1,-.5) -- (T1\x);  
                      \draw [-stealth,very thick,rounded corners] (T2\x) -- +(.1,.5) -- +(-.1,.5) -- (T2\x);}
    \ifthenelse{\x=1}{%
    \draw[-stealth,rounded corners] (T1\x) -- ($(T1\x) +(.15,0)$) -- ($(T2\x) +(.15,0)$) -- (T2\x); 
    \draw[-stealth,rounded corners] (T2\x) -- ($(T2\x) +(-.15,0)$) -- ($(T1\x) +(-.15,0)$) -- (T1\x); 
    }{} 
    \ifthenelse{\x=2}{%
    \draw[-stealth,rounded corners,dashed] (T1\x) -- ($(T1\x) +(.15,.15)$) -- node[right]{$0.01$} ($(T2\x) +(.15,0)$) -- (T2\x); 
    \draw[-stealth,rounded corners,dashed] (T2\x) -- ($(T2\x) +(-.15,-.15)$) -- node[left]{$0.01$} ($(T1\x) +(-.15,0)$) -- (T1\x);
    \draw [-stealth,rounded corners,dashed] (T1\x) -- +(.2, -.4) -- node[right]{$0.99$} +(.2,0) -- (T1\x);  
    \draw [-stealth,rounded corners,dashed] (T2\x) -- +(-.2, .4) -- node[left]{$0.99$} +(-.2,0) -- (T2\x);  
                      
    }{} 
  
}
\end{tikzpicture}
\caption{We compare two strategies (b) and (c) on graph (a). Strategy (b) patrols the targets in a deterministic loop, while strategy (c) applies randomization. Both have the same expected time lag between two consecutive visits to a target but differ in the protection significantly.}
\label{fig-renewal}
\end{figure}

Let us note that although the example of Fig.~\ref{fig-renewal} contains self-loops (which do not appear in real-world patrolling graphs), these can easily be avoided by inserting auxiliary vertices so that the demonstrated deficiency is still present.  

One may still argue that the above problem is caused just by allowing the Defender to use randomized strategies. This is a valid objection. So, the question is whether randomization is really needed, i.e., whether the Defender can achieve strictly better protection by using randomized strategies. 

Consider the graph of Fig.~\ref{fig-adeversarial} with two targets $\tau_1$ and $\tau_2$ where $\alpha(\tau_1) = 1$, $\alpha(\tau_2) =2$, and the traversal time of every edge is~$1$. 

The protection achieved by an arbitrary deterministic strategy is not better than~$8$, because the Attacker can wait until the Defender starts moving from $\tau_2$ to $v$ so that the next move selected after arriving to $v$ will be edge leading to $\tau_1$. Note that the Attacker knows the Defender's strategy and can observe its moves, and hence he can recognize this attack opportunity. Since the Defender needs at least $4$ time units to return to $\tau_2$, the expected damage caused by this attack is at least~$8$. 

The Attacker's ability to anticipate future Defender's moves can be decreased by randomization. Consider the simple strategy $\sigma_b$ of Fig.~\ref{fig-adeversarial}b, where the Defender moves from $v$ to $\tau_1$ and $\tau_2$ with probability $p$ and $1-p$, respectively. After reaching a target, the Defender returns to $v$. The optimal value of $p$ is $\frac{7}{2} - \frac{\sqrt{41}}{2} \doteq 0.3$, and the expected damage of an optimal attack is then $\doteq 7.7$. Note that the strategy $\sigma_b$ is \emph{memoryless}, i.e., its decisions depend just on the currently visited vertex.

At first glance, it is not clear whether the protection can be improved, because the probability $p$ used by $\sigma_b$ implements an optimal ``balance'' between visiting $\tau_1$ and $\tau_2$ determined by the weights of $\tau_1$ and $\tau_2$. However, consider the finite-memory strategy $\sigma_c$ of  Fig.~\ref{fig-adeversarial}c which is ``almost deterministic'' except for the moment when the Defender returns to $v$ from $\tau_2$. Here, the strategy select the next edge uniformly at random. Then, the expected damage of an optimal attack is~$6$ (the best attack opportunity is to attack $\tau_2$ right after the robot starts moving form $\tau_2$ to $v$. The expected time need to visit $\tau_2$ is then equal to $3$, yielding the expected damage~$6$). Hence, protection is increased not only by randomization, but also by appropriate use of memory.

\begin{figure}[t]\centering
\begin{tikzpicture}[x=2cm, y=1.2cm, scale=0.7,font=\small]
\foreach \x/\c/\l in {0/0/c,1/1.5/b,2/3.5/a}{%
    \coordinate (a) at (0,\c);
    \node at ($(a) +(-.8,0)$) {(\l)};
    \node [min,double] (T1\x) at (a) {$\tau_1$};
    \node [min,double] (T2\x) at ($ (a) + (2,0) $) {$\tau_2$};
    \node [min] (V\x) at ($ (a) + (1,0) $) {$v$};
    \ifthenelse{\x=2}{%
    \draw[stealth-stealth,very thick] (T1\x) -- node[above] {$ $}  (V\x);
    \draw[stealth-stealth,very thick] (V\x) -- node[above] {$ $} (T2\x);
    \node at ($ (a) +(-.3,0)$) {$1$};
    \node at ($ (a) +(2.3,0)$) {$2$};
    }{%
    \draw[stealth-stealth,very thick] (T1\x) -- (V\x);
    \draw[stealth-stealth,very thick] (V\x) -- (T2\x);
    \draw[-stealth,rounded corners] (T1\x) -- ($(T1\x) +(0,.5)$) -- ($(V\x) + (0,.5)$);  
    }
    \ifthenelse{\x=1}{%
    \draw[-stealth,rounded corners,dashed] ($(V\x) +(0,.5)$) -- ($(V\x) +(-.2,.8)$) -- node[above]{$p$} ($(T1\x) +(-.5,.8)$) -- (T1\x); 
    \draw[-stealth,rounded corners,dashed] ($(V\x) +(0,.5)$) -- ($(V\x) +(.2,.8)$) -- node[above]{$1-p$} ($(T2\x) +(.5,.8)$) -- (T2\x); 
    \draw[-stealth,rounded corners] (T2\x) -- ($(T2\x) +(0,.5)$) -- ($(V\x) + (0,.5)$);  
    }{} 
    \ifthenelse{\x=0}{%
       \draw[-stealth,rounded corners] (T2\x) -- ($(T2\x) +(0,-.5)$) -- ($(V\x) + (0,-.5)$);
       \draw[-stealth,rounded corners,dashed] ($(V\x) +(0,-.5)$) -- node[below]{$0.5$} ($(T1\x) +(0,-.5)$) -- (T1\x); 
       \draw[-stealth,rounded corners,dashed] ($(V\x) +(0,-.5)$) -- ($(V\x) +(.2,-.8)$) -- node[below]{$0.5$} ($(T2\x) +(.5,-.8)$) -- (T2\x); 
     \draw[-stealth,rounded corners] ($(V\x) + (0,.5)$) -- ($(T2\x) +(0,.5)$) -- (T2\x);   

    }{}
}
\end{tikzpicture}
\caption{On graph (a), a memoryless randomized strategy (b) is outperformed by a randomized strategy (c) with finite memory.}
\label{fig-adeversarial}
\end{figure}
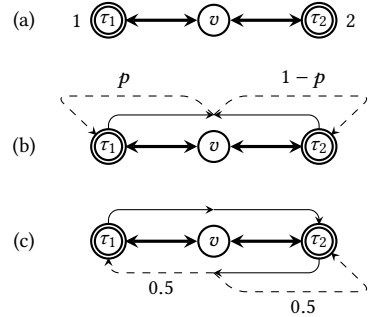

\section{Finite-memory Defender's strategies}
\label{sec-regular}

In this section, we prove that finite-memory Defender's strategies can achieve the same limit protection value as general strategies.

Let $G$ be a patrolling graph. A general Defender's strategy for $G$ (see Section~\ref{sec-defender}) depends on the whole history of moves and cannot be finitely represented. A computationally feasible subclass are \emph{finite-memory} (or \emph{regular}) strategies \cite{KL:patrol-regular,KKLR:patrol-gradient,KKMR:Regstar-UAI} where the relevant information about the history is represented by finitely memory elements assigned to each vertex. 

Formally, let $\mem\colon V\to\Nset$ be a function assigning to every vertex the number of \emph{memory elements}. The set of \emph{augmented vertices} is defined by $\dhat{V}=\{(v,m)\colon v\in V,\,1\le m\le \mem(v)\}$. We use $\dhat{v}$ to denote an augmented vertex of the form $(v,m)$ where $m \leq \mem(v)$.

A \emph{regular} Defender's strategy for $G$ is a  function $\sigma\colon\dhat{V}\to\Dist(\dhat{V})$ where $\sigma(v,m)(v',m')>0$ only if $v\to v'$. We say that $\sigma$ is \emph{unambiguous} if for all $v,v' \in V$ and $m \leq \mem(v)$ there is at most one $m'$ such that $\sigma(v,m)(v',m')>0$.

Intuitively, the Defender starts patrolling in a designated \emph{initial vertex} $v$ with \emph{initial memory element $m$}, and then traverses the vertices of $G$ and updates the memory according to $\sigma$. Hence, the current memory element represents some information about the history of visited vertices.

For every initial $\dhat{v} \in \dhat{V}$, the strategy $\sigma$ determines the probability space over the walks in the way described in Section~\ref{sec-defender}. The only difference is that the probability of $\walk(v_1,\ldots,v_n)$ where $v_1 = v$ is defined as $\sum_{\dhat{v}_2,\ldots,\dhat{v}_n} \prod_{i=1}^{n-1} \sigma(\dhat{v}_i)(\dhat{v}_{i+1})$. Here, $\dhat{v}_1 = \dhat{v}$ is the initial augmented vertex. Hence, the notion of protection value defined in Section~\ref{sec-evaluating} is applicable also to regular strategies (where $\min_v$  is replaced with $\min_{\dhat{v}}$) in~\eqref{eq:valgamma-def}).

An important question is whether regular strategies can achieve the same limit protection value as general strategies. The answer is positive, and it is proven in two steps. First, we show that there exists an \emph{optimal} Defender's strategy $\gamma$ satisfying $\Val(\gamma) = \Val$ (see Section~\ref{sec-evaluating}). Then, for arbitrarily small $\varepsilon > 0$, we prove the existence of a regular strategy $\sigma$ such that $\Val(\sigma) \leq \Val(\gamma) + \varepsilon$. 

\begin{theorem}
\label{thm-optimal}
For every patrolling graph, there exists a Defender's strategy $\gamma$ such that $\Val(\gamma) = \Val$.
\end{theorem}
\begin{proof}
By the definition of $\Val$, there exist a vertex $v$ and an infinite sequence $\Gamma = \gamma_1,\gamma_2,\ldots$ of Defender's strategies such that $\Val(\gamma_i) = \sup_{\pi} \Exp^{\gamma_i,v}[\calD^{\pi}]$ for all $i \geq 1$, and the infinite sequence $\Val(\gamma_1),\Val(\gamma_2),\ldots$ converges to $\Val$.
 
Let $\textit{Histories} = h_1,h_2,\ldots$ be a sequence where \emph{every} history occurs exactly once (without any special ordering). Let $\Gamma_0 = \Gamma$. For every $i \geq 1$, we inductively define an infinite sequence of strategies $\Gamma_i$ and a probability distribution $\gamma(h_i)$ over $V$, assuming that $\Gamma_{i-1}$ has already been defined. Let $t$ be the last vertex of $h_i$, and let $\{u_1,\ldots,u_k\}$ be the set of all immediate successors of $t$ in $G$ (i.e., $t \to u_j$ for all $j \leq k$). Since every bounded infinite sequence of real numbers contains an infinite convergent subsequence, there exists an infinite subsequence $\Gamma_i = \varrho_1,\varrho_2,\ldots$ of $\Gamma_{i-1}$ such that the sequence 
\[
   \varrho_1(h_i)(u_j),\ \varrho_2(h_i)(u_j),\ \ldots
\] 
is convergent for every $j \leq k$. We put 
\[
   \gamma(h_i)(u_j) \quad = \quad \lim_{\ell \to \infty}\ \varrho_{\ell}(h_i)(u_j)\,.
\] 
It is easy to check that $\sum_{j=1}^k \gamma(h_i)(u_j) = 1$, i.e., $\gamma(h_i)$ is indeed a distribution on $V$. Hence, the function $\gamma$ is a Defender's strategy, and we show that $\Val(\gamma) = \Val$.

For the sake of contradiction, suppose $\sup_{\pi} \Exp^{\gamma,v}[\calD^{\pi}] - \Val = \delta > 0$. Let $\varepsilon = \delta/4$. Then, there exists an Attackers strategy $\pi^*$ such that 
\begin{equation}
  \Exp^{\gamma,v}[\calD^{\pi^*}] \quad \geq \quad \Val\ +\ \delta\ -\ \varepsilon
\end{equation}
Let $\Attack(\Omega) \subseteq \Omega$ be the set of all observations $o$ such that $\pi^*(o) = \attack_\tau$ for some $\tau \in T$. For every \mbox{$o = v_1,\ldots,v_n,v_n {\to} v_{n+1} \in \Attack(\Omega)$}, let $\walk(o)$ be the set of all walks starting with $v_1,\ldots,v_{n+1}$. Observe that if $o,o' \in \Attack(\Omega)$ and $o \neq o'$, then $\walk(o) \cap \walk(o') = \emptyset$. Furthermore, for every $w \in \calW \smallsetminus \bigcup_{o \in \Attack(\Omega)} \walk(o)$, we have that $\calD^{\pi^*}(w) = 0$. Hence, we obtain
\begin{equation}
   \Exp^{\gamma,v}[\calD^{\pi^*}]  \quad = \quad \sum_{o \in \Attack(\Omega)} \prob^{\gamma,v}(\walk(o)) \cdot \Exp^{\gamma,v}[\calD^{\pi^*} \mid \walk(o)]
\label{eq:exp-decompose}
\end{equation}
where $\Exp^{\gamma,v}[\calD^{\pi^*} \mid \walk(o)]$ is the conditional expected value of $\calD^{\pi^*}$ under the condition that a walk starts with~$o$. (Note that the conditional expectation is undefined
when $\prob^{\gamma,v}(\walk(o))=0$. In that case, ``$0\cdot\textrm{undefined}$'' is interpreted as $0$.) Hence, there exists a \emph{finite} set of observations $O \subseteq \Attack(\Omega)$ such that 
\begin{equation}
   \Exp^{\gamma,v}[\calD^{\pi^*}]  \quad \leq \quad \varepsilon + \sum_{o \in O} \prob^{\gamma,v}(\walk(o)) \cdot \Exp^{\gamma,v}[\calD^{\pi^*} \mid \walk(o)]
\label{eq:exp-bound}
\end{equation}
For each $o = v_1,\ldots,v_n,v_n{\to}v_{n+1} \in O$, let $\tau_o$ be the target attacked by $\pi$ after observing $o$, and let $\calH(o)$ be the set of all histories $h$ initiated in $v_{n+1}$ such that $\tau_o$ is the last vertex of $h$ and $\tau_o$ occurs exactly once in $h$. We use $o \odot h$ to denote the history $v_1,\ldots,v_n,h$. We have that $\Exp^{\gamma,v}[\calD^{\pi^*} \mid \walk(o)]$ is equal to
\begin{equation}
   \sum_{h \in \calH(o)}  \prob^{\gamma,v}\big(\walk(o \odot h) \mid \walk(o)\big) \cdot \big(\tm(h) + \tm(v_n,v_{n+1})\big) \cdot \alpha(\tau_o)
\label{eq:exp-sum}
\end{equation}
where $\prob^{\gamma,v}\big(\walk(o \odot h) \mid \walk(o)\big)$ is the conditional probability of performing a walk starting with $o \odot h$ under the condition that a walk starting with $o$ is performed. Clearly, there exists a \emph{finite} $H(o) \subseteq \calH(o)$ such that the sum~\eqref{eq:exp-sum} decreases at most by $\varepsilon$ when $h$ ranges over $H(o)$ instead of $\calH(o)$. For short, we use $E^{\gamma,v}(o)$ to denote the sum
\begin{equation}
   \sum_{h \in H(o)}  \prob^{\gamma,v}\big(\walk(o \odot h) \mid \walk(o)\big) \cdot \big(\tm(h) + \tm(v_n,v_{n+1})\big) \cdot \alpha(\tau_o)
\label{eq:exp-sum2}
\end{equation}
Hence, 
\begin{equation}
   \Exp^{\gamma,v}[\calD^{\pi^*} \mid \walk(o)] \quad \leq \quad \varepsilon + E^{\gamma,v}(o)
\label{eq:exp-bound2}
\end{equation}
By combining~\eqref{eq:exp-bound} and~\eqref{eq:exp-bound2}, we obtain
\begin{equation}
   \Exp^{\gamma,v}[\calD^{\pi^*}]  \quad \leq \quad 2\varepsilon + \sum_{o \in O} \prob^{\gamma,v}(\walk(o)) \cdot E^{\gamma,v}(o)
\label{eq:gamma-increase}
\end{equation}
Since $\varepsilon = \delta/4$, this implies
\begin{equation}
   \sum_{o \in O} \prob^{\gamma,v}(\walk(o)) \cdot E^{\gamma,v}(o) \quad \geq \quad \Val + \varepsilon
\label{eq:valbound}
\end{equation} 
Let $H = \bigcup_{o \in O} H(o)$. Since $H$ is finite, there exists an index $\ell$ such that \emph{all} elements of $H$ appear among the first $\ell$ elements of $\textit{Histories}$. Consider the sequence $\Gamma_{\ell} = \varrho_1,\varrho_2,\ldots$ and observe that, for all $i\geq 1$,
\begin{eqnarray}
   \Val(\varrho_i) & = & \sup_{\pi} \Exp^{\varrho_i,v}[\calD^{\pi}] \nonumber \\
                   & \geq & \Exp^{\varrho_i,v}[\calD^{\pi^*}] \nonumber \\
                   & \geq & \sum_{o \in O} \prob^{\varrho_i,v}(\walk(o)) \cdot E^{\varrho_i,v}(o)
  \label{eq:Valrho}
\end{eqnarray}
Since the sequence of distributions $\varrho_1(h),\varrho_2(h),\ldots$ converges to $\gamma(h)$ for every $h \in H$, we also obtain that
\begin{equation}
   \lim_{i \to \infty } \sum_{o \in O} \prob^{\varrho_i,v}(\walk(o)) \cdot E^{\varrho_i,v}(o) \ = \ 
   \sum_{o \in O} \prob^{\gamma,v}(\walk(o)) \cdot E^{\gamma,v}(o)
\label{eq:limit}
\end{equation} 
By combining~\eqref{eq:valbound},~\eqref{eq:Valrho}, and~\eqref{eq:limit}, we obtain $\Val(\varrho_j) \geq \Val + \varepsilon/2$ for all sufficiently large $j$. This means that the sequence $\Val(\varrho_1),\Val(\varrho_2), \ldots$ does \emph{not} converge to $\Val$, and we have a contradiction. 
\end{proof}

\begin{theorem}
\label{thm-regular}
Let $G$ be a patrolling graph, and let $\Reg$ be the class of all regular strategies for $G$. Then $\inf_{\sigma \in \Reg} \Val(\sigma) = \Val$.
\end{theorem}
\begin{proof}
By Theorem~3.1, there exists a Defender's strategy $\gamma$ such that $\Val(\gamma) = \Val$. 
We show that for every $\varepsilon > 0$, there exist sufficiently large $d,\ell \in \Nset$ such that $\Val(\sigma_{d,\ell}) \leq \Val(\gamma) + \varepsilon$, where $\sigma_{\delta,\ell}$ is a regular strategy obtained by \emph{$d$-discretization} and \emph{$\ell$-folding} of $\gamma$. 

A $d$-discretization of $\gamma$ is a strategy $\gamma_d$ where for all $h \in \calH$ and $v \in V$, the following conditions are satisfied:
\begin{itemize}
  \item $\gamma_d(h)(v) = k/d$ for some $k \in \{0,\ldots,d\}$;
  \item $\gamma_d(h)(v) = 0$ iff $\gamma(h)(v) = 0$;
  \item $|\gamma_d(h)(v) - \gamma(h)(v)| \leq |V|/d$.  
\end{itemize}
Observe that a $d$-discretization of $\gamma$ exists for every $d \geq |V|$. 

The regular strategy $\sigma_{d,\ell}$ is obtained by $\ell$-folding the strategy $\gamma_d$ constructed for a sufficiently large $d$. We can view $\gamma_d$ as an infinite tree $T$ where the nodes are histories and $h \tran{x} hu$ iff $\gamma(h)(u) = x$.
Furthermore, we label each node $h$ of $T$ with the last vertex of $h$.
Since the edge probabilities range over finitely many values, the tree $T$ contains only \emph{finitely many} subtrees of height~$\ell$ up to isomorphism preserving both node and edge labels. For a given history $h = v_0,\ldots,v_n$,  let $f_h,s_h \in \Nset$ be the lexicographically smallest pair of indexes such that $f_h < s_h$, both $f_h$ and $s_h$ are integer multiples of~$\ell$, and the subtrees of height~$\ell$ rooted by $v_{f_h}$ and $v_{s_h}$ are isomorphic. If no such $f_h,s_h$ exist, we say that $h$ does not contain a folding pair.

Note that there exists a constant $c_\ell$ \emph{independent of $h$} such that every history $h$ of length at least $c_\ell$ contains a folding pair $f_h,s_h$ with both components  bounded by $c_{\ell}$. We define a strategy $\gamma_{d,\ell}$ as follows:
\begin{itemize}
\item $\gamma_{d,\ell}(h) = \gamma_{d}(h)$ for every history $h$ without a folding pair.
\item $\gamma_{d,\ell}(h) = \gamma_{d,\ell}(h')$ for every history $h$ with a folding pair $f_h,s_h$, where the $h'$ is obtained from $h = v_0,\ldots,v_n$ by deleting the subpath $v_{f_h},\ldots,v_{s_h -1}$.
\end{itemize}
Note that $\gamma_{d,\ell}$ can be equivalently represented as a \emph{finite-memory} strategy $\sigma_{d,\ell}$ where the memory elements correspond to the (finitely many) histories without a folding pair. 

It remains to show that for every $\varepsilon > 0$ there are sufficiently large $d,\ell$ such that $\Val(\gamma_{d,\ell}) \leq \Val + \varepsilon$. 

We start by observing important structural properties of the optimal strategy $\gamma$. Let $v^* \in V$ such that $\sup_{\pi} \Exp^{\gamma,v^*}[\calD^{\pi}] = \Val$. A history $h$ is \emph{$\gamma$-eligible} if $\prob^{\gamma,v^*}(\walk(h)) > 0$. For every $\gamma$-eligible history $h$, let $\gamma[h]$ be a strategy such that $\gamma[h](h') = \gamma(h \odot h')$ for every $h' \in \calH$ initiated in the last vertex $u$ of $h$ (for the other $h'$, the strategy $\gamma[h]$ is defined arbitrarily). We show that 
\begin{equation}
  \sup_{\pi} \Exp^{\gamma[h],u}[\calD^{\pi}] \quad = \quad \Val
\label{eq:subgame-val}
\end{equation}
Clearly,
\begin{equation}
   \sup_{\pi} \Exp^{\gamma[h],u}[\calD^{\pi}] \quad \geq \quad \Val
\end{equation}
for every $\gamma$-eligible~$h$ because otherwise we have a contradiction with the definition of $\Val$. 
Now suppose $\sup_{\pi} \Exp^{\gamma[h],u}[\calD^{\pi}] > \Val$ for some $\gamma$-eligible~$h=v_1,\ldots,v_n$. We show that then there is an Attacker's strategy $\pi$ such that $\Exp^{\gamma,v^*}[\calD^{\pi}] > \Val$, which contradicts the optimality of~$\gamma$. Let $H$ be the set of all $\gamma$-eligible histories of the form $v_1,\ldots,v_k,t$ where $k < n$ and $t \neq v_{k+1}$. Furthermore, let $\kappa = \sup_{\pi} \Exp^{\gamma[h],u}[\calD^{\pi}] -\Val$, and let $\delta > 0$  be a constant satisfying
\begin{equation*}
   \delta \cdot \big(1-\prob^{\gamma,v^*}(\walk(h))\big) \quad \leq \quad    \frac{\prob^{\gamma,v^*}(\walk(h)) \cdot \kappa}{3}  
\end{equation*}
For every $h' \in H$, let $\pi_{h'}$ be an Attacker's strategy such that $\Exp^{\gamma[h'],u}[\calD^{\pi_{h'}}] \geq \Val - \delta$, and we also fix an Attacker's strategy $\pi_h$ such that $\Exp^{\gamma[h],v_n}[\calD^{\pi_{h'}}] \geq \Val + \frac{2}{3}\kappa$. The strategy $\pi$ is defined as follows. For every $\hbar \in H \cup \{h\}$ and every observation of the form $h'', u\to t$ such that $h''t$ is \mbox{$\gamma[\hbar]$-eligible} and $\pi_{\hbar}(h'', u{\to} t) = \attack_\tau$ for some $\tau \in T$, we put $\pi(\hbar \odot h'', u {\to} t) = \attack_\tau$. Thus, we obtain $\Exp^{\gamma,v^*}[\calD^{\pi}] \geq \Val + \kappa/3$.

For every target $\tau$, let $\pi[\tau]$ be the Attacker's strategy where \mbox{$\pi[\sigma](u,u\to v) = \attack_\tau$} for all $u,v \in V$, i.e., $\pi[\tau]$ attacks $\tau$ right after the Defender starts its walk. An immediate consequence of~\eqref{eq:subgame-val} is that for every $\gamma$-eligible history $h$ ending in a vertex $u$ and every target~$\tau$, we have that 
\begin{equation}
    \Exp^{\gamma[h],u}[\calD^{\pi[\tau]}] \quad \leq \quad \Val
\label{eq:immediate-val}
\end{equation}
Let $R_{\tau}$ be a function assigning to every walk $w = v_0,v_1,\ldots$ the least~$n$ such that $v_n = \tau$. If there is no such $n$, we put $R_{\tau}(w) = \infty$. Since $R_\tau \leq \calD^{\pi[\tau]}$, we have that
\begin{equation}
   \prob^{\gamma[h],u}(R_{\tau} \geq 2\Val) \quad \leq \quad \frac{1}{2}
\label{eq:Markov}
\end{equation} 
by~\eqref{eq:immediate-val} and Markov inequality.

Note that~\eqref{eq:Markov} holds for \emph{all} $\gamma$-eligible $h$ and $\tau$. Hence, we have that
\begin{equation}
   \prob^{\gamma[h],u}(R_{\tau} \geq i \cdot 2\Val) \quad \leq \quad \frac{1}{2^i}
\label{eq:R-exp-decay}
\end{equation} 
because $R_{\tau} \geq i \cdot 2\Val$ requires success of $i$ consecutive independent experiments where each experiment succeeds with probability bounded by $1/2$. 

Now we show that for every $\varepsilon > 0$, there exist sufficiently large $d,\ell$ such that $\Val(\gamma_{d,\ell}) \leq \Val(\gamma) + \varepsilon$, which proves our theorem. 

For the rest of this proof, we fix $\varepsilon >0$. Furthermore, we fix $k \in \Nset$ satisfying 
\begin{equation}
  \sum_{i=2}^\infty \frac{1}{2^{k(i-1)}} \cdot i \cdot k \cdot 2\Val \cdot \tm_{\max} \cdot \alpha_{\max} \quad \leq \quad \varepsilon/4
\label{eq:k-def}
\end{equation}
Here $\tm_{\max} = \max\{\tm(e) \mid e \in E\}$ and $\alpha_{\max} =\max\{\alpha(\tau) \mid \tau \in T\}$.
Note that such a $k$ exists because the above sum converges to $0$ as $k \to \infty$. For every $i \in \Nset$, let $\calI_i$ be the set of all integers $j$ satisfying  
\[
   (i-1) \cdot k\cdot 2\Val \quad \leq \quad j \quad < \quad  i \cdot k\cdot 2\Val\,.
\]
We have that 
\begin{equation}
   \Exp^{\gamma[h],u}[\calD^{\pi[\tau]}] =
   \sum_{i=1}^\infty  \Exp^{\gamma[h],u}[\calD^{\pi[\tau]} \mid R_\tau \in \calI_i] \cdot
                      \prob^{\gamma[h],u}[R_\tau \in \calI_i] 
\label{eq:exp-conditional}  
\end{equation}
Observe 
\begin{eqnarray}
   \Exp^{\gamma[h],u}[\calD^{\pi[\tau]} \mid R_\tau \in \calI_i] & \leq & i \cdot k \cdot 2\Val \cdot \tm_{\max} \cdot \alpha_{\max} \label{eq:bound1}\\
   \prob^{\gamma[h],u}[R_\tau \in \calI_i] & \leq & \frac{1}{2^{k(i-1)}} \label{eq:bound2}
\end{eqnarray}
Inequality~\eqref{eq:bound1} is trivial, and~\eqref{eq:bound2} follows from~\eqref{eq:R-exp-decay}. Using~\eqref{eq:k-def}, we obtain
\begin{eqnarray}
   && \sum_{i=2}^\infty  \Exp^{\gamma[h],u}[\calD^{\pi[\tau]} \mid R_\tau \in \calI_i] \cdot \prob^{\gamma[h],u}[R_\tau \in \calI_i]\nonumber \\
   & \leq & \sum_{i=2}^\infty  i \cdot k \cdot 2\Val \cdot \tm_{\max} \cdot \alpha_{\max} \cdot \frac{1}{2^{k(i-1)}} \quad \leq \quad \varepsilon/4
   \label{eq:limit}
\end{eqnarray}

Consider a strategy $\gamma_d$ where $d \geq |V|$. Then every $\gamma$-eligible history is \mbox{$\gamma_d$-eligible}, and vice-versa. Furthermore, for every $\delta > 0$, there exists a sufficiently large $d$ such that, for all $h$, $\tau$, and $i$, 
\begin{eqnarray*}
  \Exp^{\gamma_d[h],u}[\calD^{\pi[\tau]} \mid R_\tau \in \calI_i] & \leq & 
         \Exp^{\gamma[h],u}[\calD^{\pi[\tau]} \mid R_\tau \in \calI_i] + \delta\\
  \prob^{\gamma_d[h],u}[R_\tau \in \calI_i] & \leq & (1/2 + \delta)^{k \cdot (i-1)}
\end{eqnarray*}
Consequently, we can fix a sufficiently large $d$ such that the values of
\begin{eqnarray*}
  & \Exp^{\gamma[h],u}[\calD^{\pi[\tau]} \mid R_\tau \in \calI_1] \cdot \prob^{\gamma[h],u}[R_\tau \in \calI_1],\\ 
  & \sum_{i=2}^\infty  \Exp^{\gamma[h],u}[\calD^{\pi[\tau]} \mid R_\tau \in \calI_i] \cdot \prob^{\gamma[h],u}[R_\tau \in \calI_i]
\end{eqnarray*}
increase at most by $\varepsilon/4$ when $\gamma$ is replaced with $\gamma_d$. Then,
\begin{eqnarray}
  \Exp^{\gamma_d[h],u}[\calD^{\pi[\tau]}]\cdot \prob^{\gamma_d[h],u}[R_\tau \in \calI_1] & \leq & \Val + \frac{\varepsilon}{4}\nonumber\\
  \sum_{i=2}^\infty  \Exp^{\gamma_d[h],u}[\calD^{\pi[\tau]} \mid R_\tau \in \calI_i] \cdot \prob^{\gamma_d[h],u}[R_\tau \in \calI_i] & \leq & \frac{\varepsilon}{2}
\label{eq:disc-bounds}
\end{eqnarray}

Now consider a strategy $\gamma_{d,\ell}$ where $\ell = k \cdot 2\Val$, and let $h$ be a \mbox{$\gamma_{d,\ell}$-eligible} history. Then $\Exp^{\gamma_{d,\ell}[h],u}[\calD^{\pi[\tau]}]$ is equal to
\begin{equation*}
     \sum_{i=1}^\infty  \Exp^{\gamma_{d,\ell}[h],u}[\calD^{\pi[\tau]} \mid R_\tau \in \calI_i] \cdot \prob^{\gamma_{d,\ell}[h],u}[R_\tau \in \calI_i]
\end{equation*}
By definition of $\ell$-folding, for every $i \geq 1$ there is a $\gamma_d$-eligible history $h'$ without a folding pair such that 
\begin{equation}
  \Exp^{\gamma_{d,\ell}[h],u}[\calD^{\pi[\tau]} \mid R_\tau \in \calI_i] \ = \ 
  \Exp^{\gamma_{d}[h'],u}[\calD^{\pi[\tau]} \mid R_\tau \in \calI_1]
\label{eq:reg1}
\end{equation}
Furthermore,
\begin{equation}
   \prob^{\gamma_{d,\ell}[h],u}[R_\tau \in \calI_i] \ \leq \ (1-\xi)^{k(i-1)}
\label{eq:reg2}
\end{equation}
where $\xi$ is the maximal $\prob^{\gamma_d[h'],u}[R_\tau \in \calI_1]$ such that $h'$ is a \mbox{$\gamma_d$-eligible} history $h'$ without a folding pair . 

Due to~\eqref{eq:reg1} and~\eqref{eq:reg2}, the bounds of~\eqref{eq:disc-bounds} remain valid when $\gamma_d$ is replaced with $\gamma_{d,\ell}$ and $h$ is a \mbox{$\gamma_{d,\ell}$-eligible} history. Hence, 
\begin{equation}
 \Exp^{\gamma_{d,\ell}[h],u}[\calD^{\pi[\tau]}] \leq \Val + \varepsilon
\label{eq:reg3}
\end{equation}
Since~\eqref{eq:reg3} holds for every \mbox{$\gamma_{d,\ell}$-eligible} $h$ and every target $\tau$, we obtain $\Val(\gamma_{d,\ell}) \leq \Val + \varepsilon$.
\end{proof}

\section{Strategy Synthesis Algorithm}
\label{sec-synthesis}

In this section, we design an efficient algorithm synthesizing a regular Defender's strategy for a given patrolling graph capable of balancing the trade-off between maximizing $\Val(\sigma)$ and minimizing the tail probabilities $P_c(\Val(\sigma))$.

\subsection{Computing $\Val(\sigma)$}
\label{sec-eval}
First, we show how to compute $\Val(\sigma)$ for a given regular strategy $\sigma$.
Let $G$ be a patrolling graph and $\sigma$ a regular strategy for~$G$.
Let $\dhat{E}$ be the set of all $(\dhat{u},\dhat{v}) \in \dhat{V}\times\dhat{V}$ such that $\sigma(\dhat{u})(\dhat{v})>0$,
\ie, $\dhat{E}$ is the set of \emph{augmented edges} used by $\sigma$.
For every target $\tau$, let $\pi[\tau]$ be the Attacker strategy where for all $u,v \in V$ we have that \mbox{$\pi[\tau](u,u\to v) = \attack_\tau$} , i.e., $\pi[\tau]$ attacks $\tau$ immediately after the Defender starts its walk. 

For every $\dhat{e}=(\dhat{u},\dhat{v})\in\dhat{E}$ and $\tau\in T$, let $\calL_{\tau,\dhat{e}}$ be the expected damage caused by an attack at $\tau$ scheduled right after the Defender starts traversing $\dhat{e}$, i.e.,
\begin{equation}
   \calL_{\tau,\dhat{e}}  \quad = \quad \Exp^{\sigma,\dhat{u}}
   \big[\calD^{\pi[\tau]}\mid walk(\dhat{e})\big]\,.
\end{equation}
Hence, $\calL_{\tau,\dhat{e}}$ is the conditional expected value of $\calD^{\pi[\tau]}$ under the condition that the Defender's walk starts by traversing $\dhat{e}$.

Consider the directed graph $\dhat{G}=(\dhat{V},\dhat{E})$, and let $\calB$ denote the set of all \emph{bottom} strongly connected components of $\dhat{G}$. Let

\begin{equation} \label{E:Ldef}
    \calL(\sigma) \quad = \quad \min_{B \in \calB} \ \max_{\tau \in T} \ \max_{\dhat{e} \in E(B)} \ 
    \calL_{\tau,\dhat{e}}
\end{equation}
where $E(B)=\dhat{E}\cap(B\times B)$ is the set of augmented edges in the component $B$ used by $\sigma$.
We have the following:

\begin{theorem}
\label{thm-bound}
   Let $\sigma$ be a regular strategy for a patrolling graph~$G$.
   Then $\Val(\sigma) \leq   \calL(\sigma)$. If $\sigma$ is unambiguous, then $\Val(\sigma) =   \calL(\sigma)$.
\end{theorem}
\begin{proof}
For purposes of this proof, we need to introduce several notions. An \emph{augmented walk} is an infinite sequence $\dhat{v}_1,\dhat{v}_2,\dots$ such that $v_i\to v_{i+1}$ for all $i$.
The set of all augmented walks is denoted $\dhat{\calW}$. An \emph{augmented history} is a non-empty finite prefix of an augmented walk. The set of all augmented histories is denoted by $\dhat{\calH}$,
and for $\dhat{h}\in\dhat{\calH}$, $\walk(\dhat{h})$ denotes the set of all
augmented walks starting with $\dhat{h}$.

Furthermore, for every regular strategy $\sigma$ and every initial $\dhat{v}\in\dhat{V}$, we define the probability space $(\dhat\calW,\dhat\calF,\dhat{\prob}^{\sigma,\dhat{v}})$ over the augmented walks in the expected way. 

For an arbitrary walk $w=v_1,v_2,\dots$, let $Aug(w)$ denote the set of all augmented walks
of the form $\dhat{v}_1,\dhat{v}_2,\dots$.
Observe that for every measurable event $F\in\calF$, writing $\dhat{F}=\bigcup_{w\in F}Aug(w)$,
we have $\dhat{F}\in\dhat\calF$ and $\dhat{\prob}^{\sigma,\dhat{v}}(\dhat{F})=\prob^{\sigma,\dhat{v}}(F)$. Hence, to simplify our notation, we write $\prob$ instead of $\dhat\prob$.
The extension of the random variable $\calD^\pi$ to augmented walks is straightforward.

Let $\sigma$ be a regular Defender's strategy. We prove $\Val(\sigma) \leq   \calL(\sigma)$. Recall 

\begin{eqnarray*}
   \Val(\sigma) & = & \newmin_{\dhat{v}\in \dhat{V}} \sup_{\pi}\ \Exp^{\sigma,\dhat{v}}[\calD^{\pi}]\\
   \calL(\sigma) & = & \min_{B \in \calB} \ \max_{\tau \in T} \ \max_{\dhat{e} \in E(B)} \ 
    \calL_{\tau,\dhat{e}}
\end{eqnarray*}

Let $B^*\in\calB$ be a bottom strongly connected component achieving the above minimum.
We show that
\begin{equation*}
   \newmin_{\dhat{v}\in \dhat{V}} \sup_{\pi}\ \Exp^{\sigma,\dhat{v}}[\calD^{\pi}]
   \leq \max_{\tau \in T} \ \max_{\dhat{e} \in E(B^*)} \ \calL_{\tau,\dhat{e}}.
\end{equation*}
Denoting the right-hand side by $M$, we show that there is $\dhat{v}\in\dhat{V}$ such that
$\sup_{\pi}\ \Exp^{\sigma,\dhat{v}}[\calD^{\pi}]\leq M$.
Choose an arbitrary $\dhat{v}\in B^*$ and assume, for the sake of contradiction, that
$\sup_{\pi}\ \Exp^{\sigma,\dhat{v}}[\calD^{\pi}]> M$.
Hence, there is an Attacker's strategy $\pi$ such that
$\Exp^{\sigma,\dhat{v}}[\calD^{\pi}]> M$.

Now we decompose the expectation $\Exp^{\sigma,\dhat{v}}[\calD^{\pi}]$ according to 
\emph{augmented observations}, i.e., sequences 
$\dhat{o}=\dhat{v}_1,\ldots, \dhat{v}_n, \dhat{v}_n {\rightarrow} \dhat{v}_{n+1}$,
where $v_1,\ldots, v_{n+1}$ is a path in~$G$.
Given $\dhat{o}=\dhat{v}_1,\ldots, \dhat{v}_n, \dhat{v}_n {\rightarrow} \dhat{v}_{n+1}$,
we use $o$ to denote the ``unaugmented'' observation $v_1,\ldots, v_n, v_n {\rightarrow} v_{n+1}$.
Let $\dhat{\Omega}$ be the set of all augmented observation, and
let $\Attack(\dhat{\Omega}) {\subseteq} \dhat{\Omega}$ be the set of all $\dhat{o}$ where
$\pi(o)\neq\wait$. For every \mbox{$\dhat{o} = \dhat{v}_1,\ldots,\dhat{v}_n,\dhat{v}_n {\to} \dhat{v}_{n+1}$} in $\Attack(\dhat{\Omega})$, let $\walk(\dhat{o})$ be the set of all augmented walks starting with $\dhat{v}_1,\ldots,\dhat{v}_{n+1}$. Note that if $\dhat{o},\dhat{o}' \in \Attack(\dhat{\Omega})$ and $\dhat{o} \neq \dhat{o}'$, then $\walk(\dhat{o}) \cap \walk(\dhat{o}') = \emptyset$. Furthermore, for every $\dhat{w} \in \dhat{\calW} \setminus \bigcup_{\dhat{o} \in \Attack(\dhat{\Omega})} \walk(\dhat{o})$, we have that $\calD^{\pi}(\dhat{w}) = 0$. Hence, we obtain
\begin{equation*}
    \Exp^{\sigma,\dhat{v}}[\calD^{\pi}]=\sum_{\dhat{o} \in \Attack(\dhat{\Omega})} \prob^{\sigma,\dhat{v}}(\walk(\dhat{o})) \cdot \Exp^{\sigma,\dhat{v}}[\calD^{\pi} \mid \walk(\dhat{o})].
\end{equation*}

Since
\begin{equation*}
    \sum_{\dhat{o} \in \Attack(\dhat{\Omega})} \prob^{\sigma,\dhat{v}}(\walk(\dhat{o}))\leq 1,
\end{equation*}
there is $\dhat{o}=\dhat{v}_1,\ldots,\dhat{v}_n,\dhat{v}_n{\to}\dhat{v}_{n+1} \in \Attack(\dhat{\Omega})$
such that
\begin{equation*}
    \prob^{\sigma,\dhat{v}}(\walk(\dhat{o}))>0
    \quad\hbox{and}\quad\Exp^{\sigma,\dhat{v}}[\calD^{\pi} \mid \walk(\dhat{o})]>M.
\end{equation*}
Let $\dhat{e}=(\dhat{v}_n,\dhat{v}_{n+1})$, and
let $\tau\in T$ be the attacked target (\ie, $\pi(o)=\attack_\tau$).
Since $\dhat{v}=\dhat{v}_1\in B^*$ and
$\prob^{\sigma,\dhat{v}}(\walk(\dhat{o}))>0$, it holds that $\dhat{e}\in E(B^*)$.
Since $\sigma$ is regular, its (randomized) behavior after $\dhat{e}$ is always the same,
regardless of the traversed history. Therefore,
\begin{equation*}
 \Exp^{\sigma,\dhat{v}}[\calD^{\pi} \mid \walk(\dhat{o})]=
  \Exp^{\sigma,\dhat{v}_{n}}[\calD^{\pi[\tau]} \mid \walk(\dhat{e})]=\calL_{\tau,\dhat{e}}\,.
\end{equation*}
Hence, $\calL_{\tau,\dhat{e}} > M$, which contradicts the definition of $M$.

Now assume that $\sigma$ is unambiguous.  We prove that $\Val(\sigma) \geq \calL(\sigma)$.
So, assume that  $\dhat{v}\in \dhat{V}$ achieves the minimum in the definition
of $\Val(\sigma)$. We construct an Attacker's strategy $\pi$ such that
$\Exp^{\sigma,\dhat{v}}[\calD^{\pi}] \geq \calL(\sigma)$.
For this purpose, let $\varphi\colon\calB\rightarrow T$ and $\psi\colon\calB\rightarrow\dhat{E}$
be such that for every $B\in\calB$, $\psi(B)\in E(B)$ and the choice $\tau=\varphi(B),\dhat{e}=\psi(B)$
achieves the maximal $\calL_{\tau,\dhat{e}}$ (see the definition of $\calL(\sigma)$). In particular,
for every $B\in\calB$, we have
\begin{equation} \label{eq:thm41b1}
    \calL_{\varphi(B),\psi(B)}\geq \calL(\sigma).
\end{equation}
Note that the condition $\psi(B)\in E(B)$ for every $B\in\calB$ implies
that $\psi$ is injective. Moreover, let $\calE$ denote the range of $\psi$
(thus, $\psi\colon\calB\rightarrow\calE$ is a bijection) and for every $\dhat{e}\in\calE$,
let $\tau(\dhat{e})$ denote the target $\varphi(\psi^{-1}(\dhat{e}))$.
Using this notation, \eqref{eq:thm41b1} becomes
\begin{equation} \label{eq:thm41b2}
    \calL_{\tau(\dhat{e}),\dhat{e}}\geq \calL(\sigma)
\end{equation}
for every $\dhat{e}\in\calE$.

For every observation $o\in\Omega$ s.t.{} $\prob^{\sigma,\dhat{v}}(walk(o))>0$,
there is exactly one augmented observation of the form $\dhat{o}$ such that
$\prob^{\sigma,\dhat{v}}(walk(\dhat{o}))>0$ (this follows by a trivial induction
on the length of $o$).
In the rest of this proof, for every $o\in\Omega$, the symbol
$\dhat{o}$ denotes the unique augmented observation satisfying the above.
Now, for each $o\in\Omega$, we define $\pi(o)$ as follows:
Let $\dhat{o}=\dhat{v}_1,\ldots, \dhat{v}_n, \dhat{v}_n {\rightarrow} \dhat{v}_{n+1}$.
If none of the augmented edges $(\dhat{v}_i,\dhat{v}_{i+1})$ for $1\leq i<n$ is in $\calE$
and the last augmented edge $\dhat{e}=(\dhat{v}_n,\dhat{v}_{n+1})$ \emph{does} appear in $\calE$,
we put $\pi(o)=\attack_{\tau(\dhat{e})}$.
Otherwise, we put $\pi(o)=\wait$.
For every $\dhat{e}\in\calE$, let $\Attack(\dhat{e})$ denote the union of $\walk(\dhat{o})$
over all $o\in\Omega$ such that $\pi(o)\neq\wait$ and $\dhat{o}$ ends with $\dhat{e}$.
Since the Attacker may attack only once along any (augmented) walk,
we have $\Attack(\dhat{e}) \cap \Attack(\dhat{e}') = \emptyset$ for all
$\dhat{e},\dhat{e}'\in\calE$ such that $\dhat{e} \neq \dhat{e}'$.
Since $\calD^\pi$ is non-negative, we obtain
\begin{equation} \label{eq:thm41b3}
    \Exp^{\sigma,\dhat{v}}[\calD^{\pi}] \geq \sum_{\hat{e} \in \calE}
    \prob^{\sigma,\dhat{v}}(\Attack(\dhat{e}))
        \cdot \Exp^{\sigma,\dhat{v}}[\calD^{\pi} \mid \Attack(\dhat{e})].
\end{equation}
Since $\sigma$ is regular and, for every $\dhat{e}\in\calE$,
$\pi$ always attacks the same target $\tau(\dhat{e})$ when the Defender starts traversing $\dhat{e}$
(regardless of the previous history), we have
\begin{equation*}
    \Exp^{\sigma,\dhat{v}}[\calD^{\pi} \mid \Attack(\dhat{e})] = \calL_{\tau(\dhat{e}),\dhat{e}}
    \geq \calL(\sigma),
\end{equation*}
where the inequality follows from \eqref{eq:thm41b2}.
Substituting into \eqref{eq:thm41b3} yields
\begin{equation*}
    \Exp^{\sigma,\dhat{v}}[\calD^{\pi}] \geq \calL(\sigma) \cdot \sum_{\hat{e} \in \calE}
    \prob^{\sigma,\dhat{v}}(\Attack(\dhat{e})).
\end{equation*}
The desired inequality $\Exp^{\sigma,\dhat{v}}[\calD^{\pi}] \geq \calL(\sigma)$ follows from
the fact that the above sum is equal to $1$ (this follows by applying basic results of finite Markov chain theory; the Defender almost surely visits some bottom strongly connected component $B\in\calB$
and there it almost surely traverses every edge infinitely often. In particular, the Defender almost surely visits the unique edge $\dhat{e}\in \calE\cap E(B)$).
\end{proof}

Our proof of Theorem~\ref{thm-bound} reveals that the Attacker can cause the expected damage equal to $\calL(\sigma)$ if it can observe the memory updates performed by the Defender.
If $\sigma$ is unambiguous, then the Attacker can determine the memory updates just by observing the history of Defender's moves.
However, if the memory updates are randomized, the Attacker needs to access the Defender's internal data structures during a patrolling walk.
Depending on a setup, this may or may not be possible.
By the worst-case paradigm of adversarial patrolling, $\calL(\sigma)$ is more appropriate than $\Val(\sigma)$ for measuring the protection achieved by~$\sigma$.
As we shall see, our strategy synthesis algorithm typically outputs unambiguous regular strategies where $\calL(\sigma) = \Val(\sigma)$.
Hence, it does not really matter whether $\calL(\sigma)$ is understood as the protection achieved by $\sigma$ or just a bound on this protection.

Fix $B \in \calB$ and $\tau \in T$. If $B$ does \emph{not} contain any augmented vertex of the form $\dhat{\tau}$, then $\calL_{\tau,\dhat{e}} = \infty$ for all $\dhat{e} \in E(B)$.
Otherwise, to every $\dhat{e}=(\dhat{u},\dhat{v}) \in E(B)$, we associate a variable $X_{\dhat{e}}$, and create a system of linear equations
\begin{equation} \label{E:system1}
    X_{\dhat{e}} \; = \; 
		\tm(u,v) + 
        \begin{cases}
           0      & \mbox{if $v = \tau$,}\\
           \sum_{\dhat{v} \to \dhat{w}} \sigma(\dhat{v})(\dhat{w}) \cdot X_{(\dhat{v},\dhat{w})} &
                    \mbox{otherwise}
        \end{cases}
\end{equation}
over $\dhat e\in E(B)$.
By a straightforward generalization of \citep[Theorem~1.3.5]{Norris:Book},
system~\eqref{E:system1} has a unique solution, equal to $(\calL_{\tau,\dhat{e}})_{\dhat{e}\in E(B)}$.

Observe that for all augmented edges $\dhat{e},\dhat{g}\in E(B)$ leading to the same augmented vertex
(say $\dhat{e}=(\dhat{u},\dhat{v})$ and $\dhat{g}=(\dhat{w},\dhat{v})$), we have
$\calL_{\tau,\dhat{g}}=\calL_{\tau,\dhat{e}}-\tm(u,v)+\tm(w,v)$.
Therefore, we may reduce the number of variables as well as equations of the system from $|E(B)|$ to $|B|$.
Indeed, to every $\dhat{v} \in B$, we assign a variable $Y_{\dhat{v}}$, and construct a system of linear equations
\begin{equation} \label{E:system2}
    Y_{\dhat{v}} \; = \;
        \begin{cases}
           0      & \mbox{if $v = \tau$,}\\
           \sum_{\dhat{v} \to \dhat{w}} \sigma(\dhat{v})(\dhat{w}) \cdot (\tm(v,w)+Y_{\dhat{w}}) &
                    \mbox{otherwise.}
        \end{cases}
\end{equation}
Then, for every $\dhat{e}=(\dhat{u},\dhat{v}) \in E(B)$,
we have $\calL_{\tau,\dhat{e}}=\tm(u,v)+y_{\dhat{v}}$, where $(y_{\dhat{v}})_{\dhat{v}\in B}$ is the unique solution of the system~\eqref{E:system2}.

\subsection{Optimization Scheme}
\label{sec-optimization}

Our strategy synthesis algorithm is based on interpreting $\calL$ as a piecewise differentiable function and applying methods of differentiable programming. 
We start from a random strategy $\sigma$, repeatedly compute $\calL(\sigma)$ and update the strategy against the direction of its gradient.

The optimization algorithm is described in Algo.~\ref{alg:optim}.
On forward pass, strategies are produced from real-valued coefficients by a \emph{Softmax} function that outputs probability distributions.
For every target $\tau$, we solve the system~\eqref{E:system2} to obtain a damage $(\calL_{\tau,\dhat e})_{\tau\in T,\dhat e\in E(B)}$.
Then, instead of hard maximum in equation~\eqref{E:Ldef} we optimize a \emph{loss} function defined by
\begin{equation} \label{E:loss}
	\text{loss} \; =\; \sum\nolimits_{\tau,\dhat e} \Phi_\varepsilon (\calL_{\tau,\dhat e})^2,
\end{equation}
where $\Phi_\varepsilon(t)=0$ for $t\in[0,m-\varepsilon m)$ and
$\Phi_\varepsilon(t)=1+(t-m)/\varepsilon m$ for $t\in[m-\varepsilon m,m]$,
in which $m=\overline{\calL}$ is the hard maximum, the bar denotes the stop-gradient operator, and $\varepsilon\in(0,1)$ is a hyperparameter.
Minimizing \emph{loss} instead of $\calL$ leads to a more efficient gradient propagation.
On top of this, we enforce the model to prefer deterministic strategies over randomized by adding an average entropy of strategies' probability distributions with a factor $\beta$.

On backward pass, the loss gradient is computed using the automatic differentiation,
we add decaying Gaussian noise and update the coefficients using Adam optimizer \citep{Adam}.

As \emph{Softmax} never produces probability distributions containing zeros,
we cut the outputs at a certain threshold (called \emph{rounding threshold}) to allow endpoint values on evaluation. Note that edges with zero probabilities are excluded from $E(B)$ which is crucial for equation~\eqref{E:system2}.

\begin{algorithm}
\caption{Strategy optimization}
\label{alg:optim}
\begin{algorithmic}
\State coefficients $\gets$ \textbf{Init}()
\For{step $\in$ steps} 
	\LineComment{Forward pass}
	\State strategy $\gets$ \textbf{Softmax}(coefficients)
	\LineComment[2]{solving linear system~\eqref{E:system2}}
	\State damage $\gets$ \textbf{Solve}(strategy)
	\State loss $\gets$ \textbf{Loss}(damage) $+$ $\beta\cdot$\textbf{Entropy}(strategy)
	\LineComment{Backward pass}
	\State coefficients.grad $\gets$ \textbf{Gradient}(loss)
	\LineComment[2]{automatic differentiation}
	\State coefficients.grad += \textbf{Noise}(step)
	\LineComment[2]{Adam optimizer's step}
	\State coefficients += \textbf{Step}(coefficients.grad, step)
	\LineComment{Strategy evaluation}
	\State strategy $\gets$ \textbf{Cutoff}(\textbf{Softmax}(coefficients))
	\State damage $\gets$ \textbf{Solve}(strategy)
	\State $\calL\gets$ \textbf{Max}(damage)
	\State \textbf{Save} $\calL$, strategy
\EndFor
\Return strategy with the smallest $\calL$
\end{algorithmic}
\end{algorithm}

\section{Experiments}
\label{sec-experiments}

We experimentally evaluate strategy synthesis algorithm on series of synthetic graphs with increasing sizes.
We perform two sets of tests. The first analyzes runtimes while the second one focuses on the achieved protection values. 

The experiments were performed on a desktop machine with Ubuntu 20.04 LTS running on
Intel\textsuperscript{\textregistered} Core\texttrademark{} i7-8700 Processor (6 cores, 12 threads) with 32GB RAM.

\subsection{Runtime Analysis}\label{sec-exp-grids}

We generate synthetic graphs with $n=10,20,\ldots,100$ vertices.
To obtain random but similarly structured graphs, we start with a grid of size $n \times n$ and choose its $n$ nodes as the vertices of our patrolling graph, half of them being targets.
All vertices are equipped with 6 memory elements.
The travel time between vertices is set to the number of edges on the shortest path in the original grid.
In the final patrolling graph, we omit those edges that have an alternative connection of at most the same length visiting another vertex.

For each $n$, we generate 10 graphs of $n$ vertices and run 10 optimization trials with 100 steps for each graph.
In Fig.~\ref{fig:grid_times}, we report statistics of average step-times in seconds aggregated by $n$.
Note than even considerably large graphs are processed in units of seconds, which confirms the applicability of our algorithm to dynamically changing environments (see Section~\ref{sec-intro}).

Recall that one optimization step consists of a forward pass (damage and loss of the current strategy), backward pass (gradient), and one more test evaluation.
For hyperparameters, we set $\varepsilon=0.3$, $\beta=0.2$, learning rate $=0.5$, cutoff threshold $=0.1$, and rounding threshold $=0.001$ (for a deeper explanation, see Section~\ref{sec-optimization}).

\begin{figure}[htb]
			\includegraphics[width=\columnwidth]{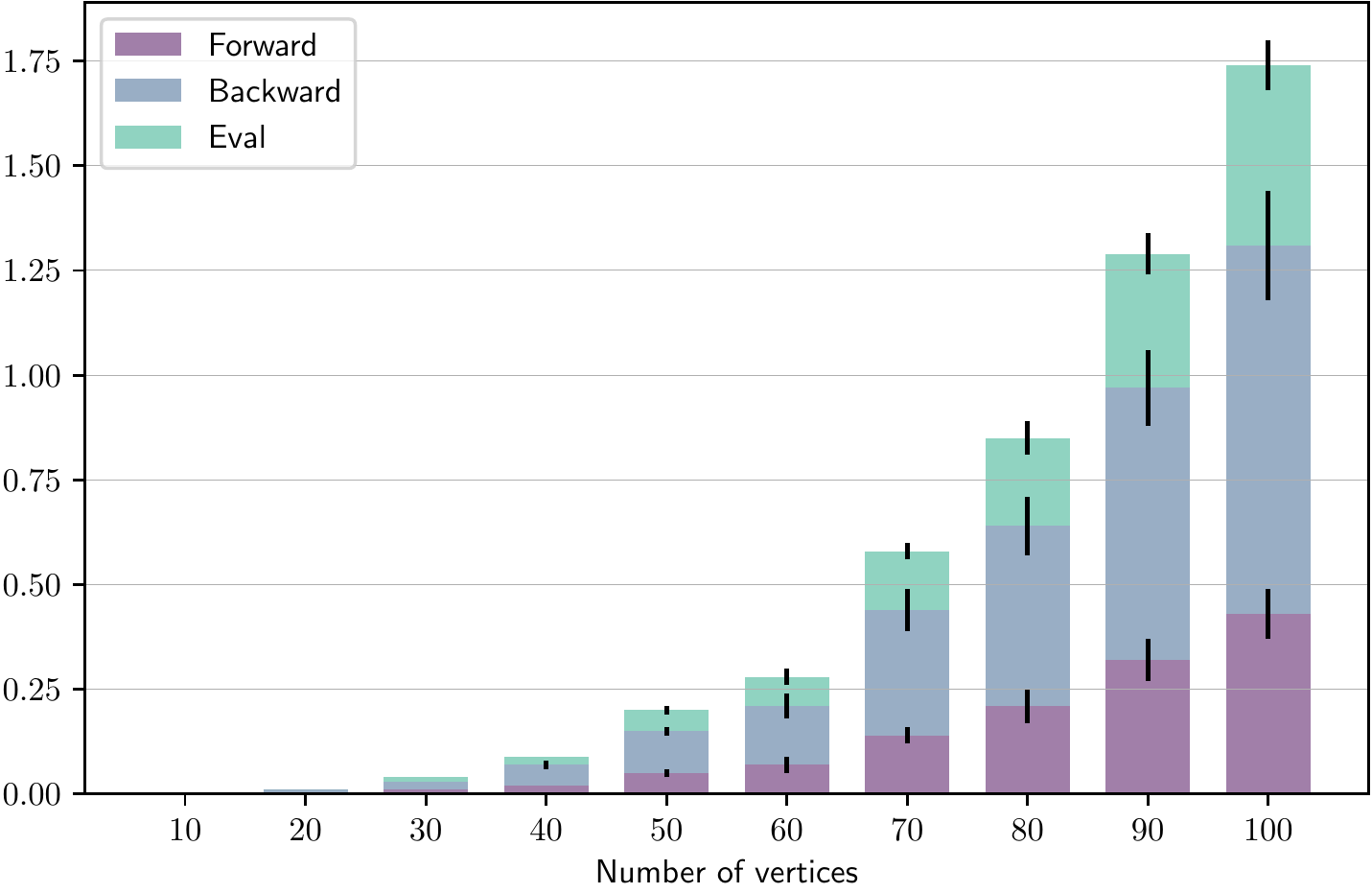}
			\caption{Runtime analysis of the strategy synthesis \mbox{algorithm} (in seconds).}
			\label{fig:grid_times}
\end{figure}

Fig.~\ref{fig:grid_convergence} shows the convergence of values during the optimization process.
Here we fix one graph for each number of vertices and run 10 trials each with 120 steps.
The colors are assigned to individual graphs, the areas show ranges of the obtained values during the optimization process.
The solid lines highlight minimal values.

\begin{figure}[htb]
        \includegraphics[width=\columnwidth]{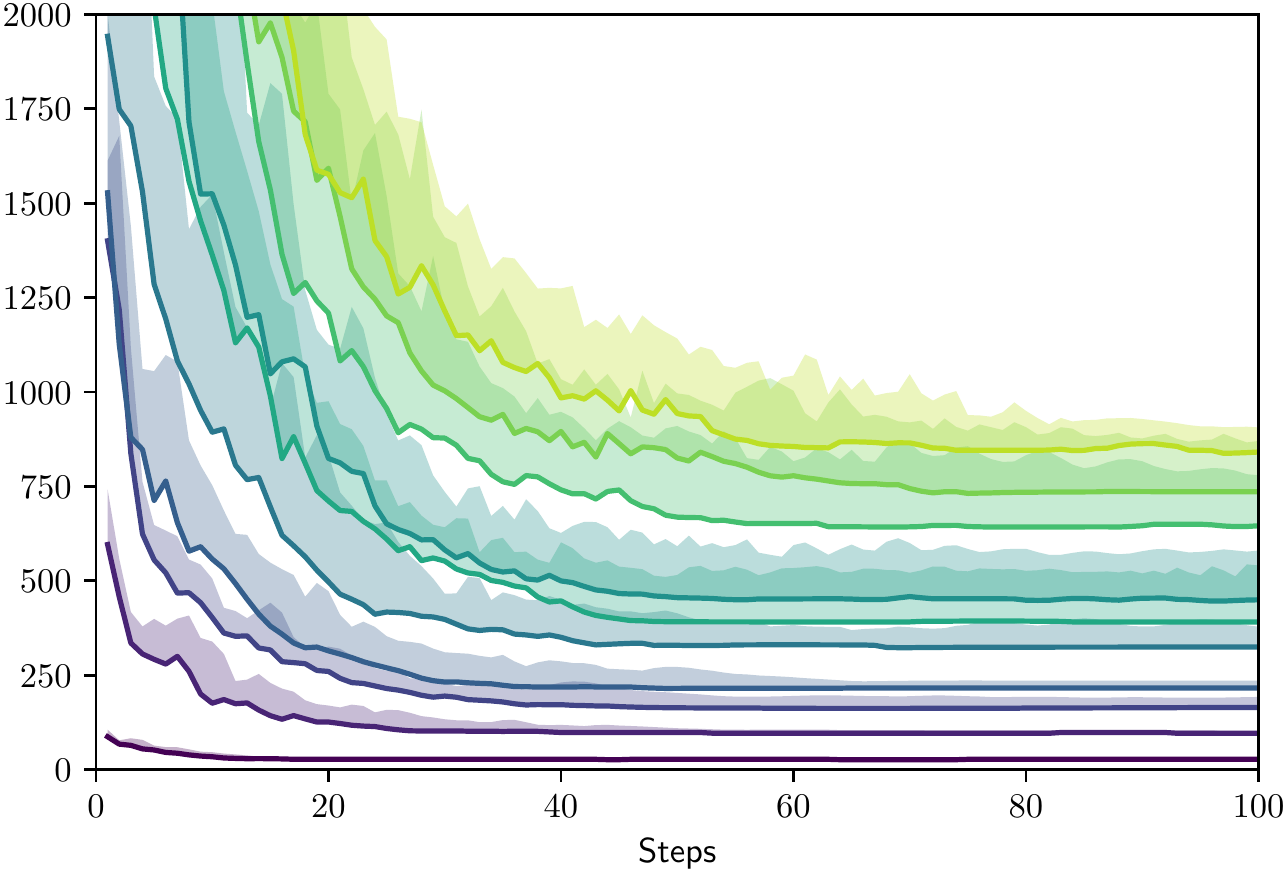}
        \caption{Values of strategies synthesized during the first 120 steps.
				The colored areas show ranges of the obtained values during the optimization process.
				The solid lines highlight the minimal values.}
        \label{fig:grid_convergence}
\end{figure}

\subsection{Patrolling Airport Gates}

One typical application of patrolling is security patrol at airports (see Section~\ref{sec-intro}).
Airport buildings have a specific tree structure of terminals with symmetric gates with a central node connecting the terminals. A terminal typically consists of pairs of gates joined by halls.
A patrolling graph for an airport with three terminals of 4, 2, and 6 gates is shown in Fig.~\ref{fig:airport}. 

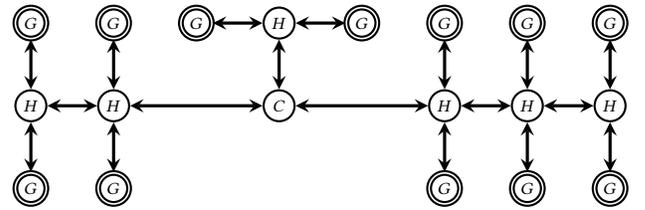
\begin{figure}[htb]
\begin{tikzpicture}[x=1.1cm, y=1.1cm,font=\scriptsize]
    \node [min] (C) at (0,0) {$C$}; 
    \foreach \x in {1,2}{%
       \ifthenelse{\x=1}{%
          \coordinate (a) at (-3,0);
          \foreach \z/\p in {0/0,1/0}{%
              \node [min] (H\x\z) at ($(a) + (\z,0)$)  {$H$};
              \node [min,double] (G0\x\z) at ($(a) + (\z,-1)$) {$G$};
              \node [min,double] (G1\x\z) at ($(a) + (\z,1)$)  {$G$};
              \draw [stealth-stealth,very thick]  (H\x\z) -- (G0\x\z);
              \draw [stealth-stealth,very thick]  (H\x\z) -- (G1\x\z);
              \ifthenelse{\z=0}{}{\draw [stealth-stealth,very thick] (H\x\z) -- (H\x\p);}
          }}{%
          \coordinate (a) at (2,0);
          \foreach \z/\p in {0/0,1/0,2/1}{%
              \node [min] (H\x\z) at ($(a) + (\z,0)$)  {$H$};
              \node [min,double] (G0\x\z) at ($(a) + (\z,-1)$) {$G$};
              \node [min,double] (G1\x\z) at ($(a) + (\z,1)$)  {$G$};    
              \draw [stealth-stealth,very thick]  (H\x\z) -- (G0\x\z);
              \draw [stealth-stealth,very thick]  (H\x\z) -- (G1\x\z);
              \ifthenelse{\z=0}{}{\draw [stealth-stealth,very thick] (H\x\z) -- (H\x\p);}
          }}   
      }
   \node [min] (H0) at (0,1) {$H$};
   \node [min, double] (G0) at (-1,1) {$G$};
   \node [min, double] (G1) at (1,1) {$G$};
   \draw [stealth-stealth,very thick]  (H0) -- (G0);
   \draw [stealth-stealth,very thick]  (H0) -- (G1);
   \draw [stealth-stealth,very thick]  (C)  -- (H20);
   \draw [stealth-stealth,very thick]  (C)  -- (H11);
   \draw [stealth-stealth,very thick]  (C)  -- (H0);
\end{tikzpicture}
\caption{A patrolling graph for an airport with 3 terminals.}
\label{fig:airport}
\end{figure}

We generate a sequence of random airport graphs with 3~terminals and an increasing number of gates determined randomly. Hence, an airport graph with $n$ gates has $n/2$ halls and exactly one central node.
The gates are targets and have exactly one memory element, while halls and the central node are non-target vertices with $4$~memory elements each (thus, the strategy can ``remember'' the previously visited vertex).
The costs of all targets are equal to $1$, and 
all edges have the same traversal time~$1$. 

Since the target costs are the same, we can estimate $\Val$ (the achievable protection) by the length of the shortest cycle visiting all targets, which is $2 (|V|-1)$ where $V$ is the set of vertices. Note that automatic synthesis of a regular strategy with comparable protection is tricky---the synthesis algorithm must ``discover'' the relevance of the previously visited vertex and design the memory updates accordingly.  

We use the same hyperparameters as above.
For each airport, we synthesized 30 strategies, each in 500 steps of iterations. In Fig.~\ref{fig:airport_norm}, we show the protection values of the synthesized strategies for increasing number of vertices \emph{normalized by the baseline $2(|V|-1)$}. For larger $|V|$, 
the optimization converges to locally optimal randomized strategies with worse protection than the deterministic-loop strategy. In particular, for $|V| = 91$, the protection achieved by the constructed strategy is about $33\%$ worse than the baseline on average (with best found strategy loosing $20\%$ above the baseline).

\begin{figure}[htb]
	\includegraphics[width=\columnwidth]{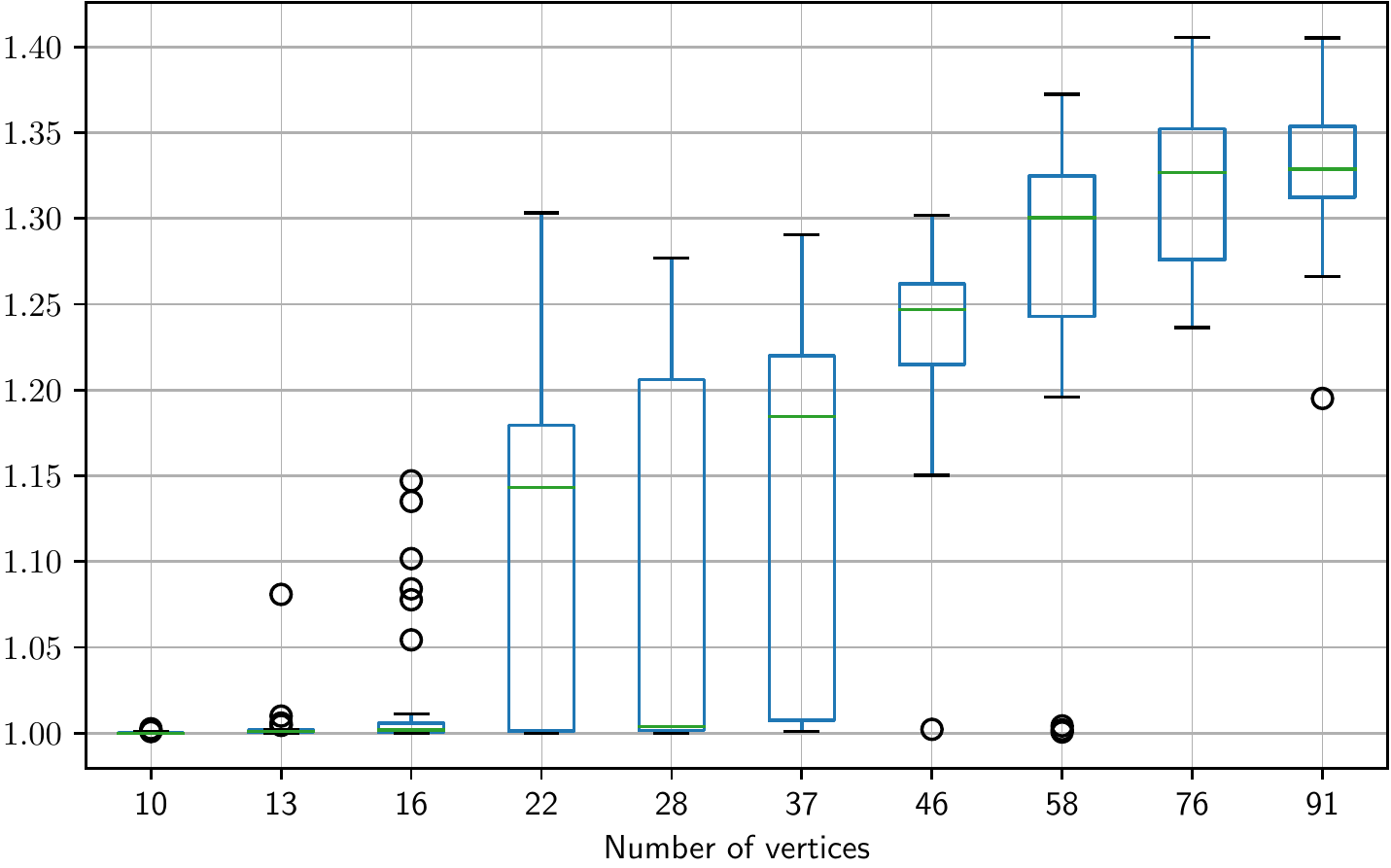}
	\caption{Normalized values of strategies synthesized for airport graphs with increasing number of vertices. The normalization baseline is the shortest cycle visiting all targets with length $2 (v-1)$ where $v$ stands for the number of vertices.}
	\label{fig:airport_norm}
\end{figure}

\section{Conclusion}
\label{sec-concl}

The outcomes show that high-quality Defender's strategies are computed very quickly even for instances of considerable size. Hence, our algorithm can also be used to re-compute a Defender's strategy dynamically when the patrolling scenario changes.

The problem encountered in our experiments is the existence of locally optimal randomized strategies where the optimization loop gets stuck. An interesting question is whether this problem can be overcome by tuning the parameters of gradient descent or by constructing the initial seeds in a more sophisticated way. 

A natural continuation of our study is extending the presented results to scenarios with multiple Defenders and Attackers.

\section*{Acknowledgements}
Research was sponsored by the Army Research Office and was accomplished under
Grant Number W911NF-21-1-0189.

V\'{\i}t Musil is also supported from Operational Programme Research, Development and
Education - Project Postdoc2MUNI (No.~CZ.02.2.69/0.0/0.0/18\_053/0016952).

\paragraph*{Disclaimer}
The views and conclusions contained in this document are those of the authors
and should not be interpreted as representing the official policies, either
expressed or implied, of the Army Research Office or the U.S.\ Government. The
U.S.\ Government is authorized to reproduce and distribute reprints for
Government purposes notwithstanding any copyright notation herein.

\newpage


%


\begin{thebibliography}{37}


\ifx \showCODEN    \undefined \def \showCODEN     #1{\unskip}     \fi
\ifx \showDOI      \undefined \def \showDOI       #1{#1}\fi
\ifx \showISBNx    \undefined \def \showISBNx     #1{\unskip}     \fi
\ifx \showISBNxiii \undefined \def \showISBNxiii  #1{\unskip}     \fi
\ifx \showISSN     \undefined \def \showISSN      #1{\unskip}     \fi
\ifx \showLCCN     \undefined \def \showLCCN      #1{\unskip}     \fi
\ifx \shownote     \undefined \def \shownote      #1{#1}          \fi
\ifx \showarticletitle \undefined \def \showarticletitle #1{#1}   \fi
\ifx \showURL      \undefined \def \showURL       {\relax}        \fi
\providecommand\bibfield[2]{#2}
\providecommand\bibinfo[2]{#2}
\providecommand\natexlab[1]{#1}
\providecommand\showeprint[2][]{arXiv:#2}

\bibitem[\protect\citeauthoryear{Agmon, Kraus, and Kaminka}{Agmon
  et~al\mbox{.}}{2008a}]%
        {AKK:multi-robot-perimeter-adversarial}
\bibfield{author}{\bibinfo{person}{N. Agmon}, \bibinfo{person}{S. Kraus}, {and}
  \bibinfo{person}{G. Kaminka}.} \bibinfo{year}{2008}\natexlab{a}.
\newblock \showarticletitle{Multi-Robot Perimeter Patrol in Adversarial
  Settings}. In \bibinfo{booktitle}{\emph{Proceedings of ICRA 2008}}.
  \bibinfo{publisher}{IEEE Computer Society Press},
  \bibinfo{pages}{2339--2345}.
\newblock


\bibitem[\protect\citeauthoryear{Agmon, Sadov, Kaminka, and Kraus}{Agmon
  et~al\mbox{.}}{2008b}]%
        {ASKK:perimeter-patrol}
\bibfield{author}{\bibinfo{person}{N. Agmon}, \bibinfo{person}{V. Sadov},
  \bibinfo{person}{G.A.{} Kaminka}, {and} \bibinfo{person}{S. Kraus}.}
  \bibinfo{year}{2008}\natexlab{b}.
\newblock \showarticletitle{The impact of adversarial knowledge on adversarial
  planning in perimeter patrol}. In \bibinfo{booktitle}{\emph{Proceedings of
  AAMAS 2008}}. \bibinfo{pages}{55--62}.
\newblock


\bibitem[\protect\citeauthoryear{Almeida, Ramalho, Santana, Tedesco, Menezes,
  Corruble, and Chevaleyr}{Almeida et~al\mbox{.}}{2004}]%
        {ARSTMCC:multi-patrolling-survey}
\bibfield{author}{\bibinfo{person}{A. Almeida}, \bibinfo{person}{G. Ramalho},
  \bibinfo{person}{H. Santana}, \bibinfo{person}{P. Tedesco},
  \bibinfo{person}{T. Menezes}, \bibinfo{person}{V. Corruble}, {and}
  \bibinfo{person}{Y. Chevaleyr}.} \bibinfo{year}{2004}\natexlab{}.
\newblock \showarticletitle{Recent Advances on Multi-Agent Patrolling}.
\newblock \bibinfo{journal}{\emph{Advances in Artificial Intelligence -- SBIA}}
   \bibinfo{volume}{3171} (\bibinfo{year}{2004}), \bibinfo{pages}{474--483}.
\newblock


\bibitem[\protect\citeauthoryear{An, Shieh, Yang, Tambe, Baldwin, DiRenzo,
  Maule, and Meyer}{An et~al\mbox{.}}{2014}]%
        {ASYTBDMM:Protect-AIMagazine}
\bibfield{author}{\bibinfo{person}{B. An}, \bibinfo{person}{E. Shieh},
  \bibinfo{person}{R. Yang}, \bibinfo{person}{M. Tambe}, \bibinfo{person}{C.
  Baldwin}, \bibinfo{person}{J. DiRenzo}, \bibinfo{person}{B. Maule}, {and}
  \bibinfo{person}{G. Meyer}.} \bibinfo{year}{2014}\natexlab{}.
\newblock \showarticletitle{Protect---A Deployed Game Theoretic System for
  Strategic Security Allocation for the {United States} Coast Guard}.
\newblock \bibinfo{journal}{\emph{AI Magazine}} \bibinfo{volume}{33},
  \bibinfo{number}{4} (\bibinfo{year}{2014}), \bibinfo{pages}{96--110}.
\newblock


\bibitem[\protect\citeauthoryear{Basilico, Gatti, and Amigoni}{Basilico
  et~al\mbox{.}}{2009a}]%
        {BGA:patrolling-arbitrary-topologies}
\bibfield{author}{\bibinfo{person}{N. Basilico}, \bibinfo{person}{N. Gatti},
  {and} \bibinfo{person}{F. Amigoni}.} \bibinfo{year}{2009}\natexlab{a}.
\newblock \showarticletitle{Leader-follower strategies for robotic patrolling
  in environments with arbitrary topologies}. In
  \bibinfo{booktitle}{\emph{Proceedings of AAMAS 2009}}.
  \bibinfo{pages}{57--64}.
\newblock


\bibitem[\protect\citeauthoryear{Basilico, Gatti, and Amigoni}{Basilico
  et~al\mbox{.}}{2012a}]%
        {BGA:large-patrol-AI}
\bibfield{author}{\bibinfo{person}{N. Basilico}, \bibinfo{person}{N. Gatti},
  {and} \bibinfo{person}{F. Amigoni}.} \bibinfo{year}{2012}\natexlab{a}.
\newblock \showarticletitle{Patrolling Security Games: Definitions and
  Algorithms for Solving Large Instances with Single Patroller and Single
  Intruder}.
\newblock \bibinfo{journal}{\emph{Artificial Inteligence}}
  \bibinfo{volume}{184--185} (\bibinfo{year}{2012}), \bibinfo{pages}{78--123}.
\newblock


\bibitem[\protect\citeauthoryear{Basilico, Gatti, Rossi, Ceppi, and
  Amigoni}{Basilico et~al\mbox{.}}{2009b}]%
        {Basilico2009-2}
\bibfield{author}{\bibinfo{person}{N. Basilico}, \bibinfo{person}{N. Gatti},
  \bibinfo{person}{T. Rossi}, \bibinfo{person}{S. Ceppi}, {and}
  \bibinfo{person}{F. Amigoni}.} \bibinfo{year}{2009}\natexlab{b}.
\newblock \showarticletitle{{Extending algorithms for mobile robot patrolling
  in the presence of adversaries to more realistic settings}}. In
  \bibinfo{booktitle}{\emph{Proccedings of WI-IAT 2009}}.
  \bibinfo{pages}{557--564}.
\newblock


\bibitem[\protect\citeauthoryear{Basilico, Gatti, and Villa}{Basilico
  et~al\mbox{.}}{2010}]%
        {Basilico2010}
\bibfield{author}{\bibinfo{person}{Nicola Basilico}, \bibinfo{person}{Nicola
  Gatti}, {and} \bibinfo{person}{Federico Villa}.}
  \bibinfo{year}{2010}\natexlab{}.
\newblock \showarticletitle{{Asynchronous Multi-Robot Patrolling against
  Intrusion in Arbitrary Topologies}}. In \bibinfo{booktitle}{\emph{Proceedings
  of AAAI 2010}}.
\newblock


\bibitem[\protect\citeauthoryear{Basilico, Nittis, and Gatti}{Basilico
  et~al\mbox{.}}{2012b}]%
        {BNG:patrolling-alarm}
\bibfield{author}{\bibinfo{person}{N. Basilico}, \bibinfo{person}{G.~De
  Nittis}, {and} \bibinfo{person}{N. Gatti}.} \bibinfo{year}{2012}\natexlab{b}.
\newblock \showarticletitle{A Security Game Combining Patrolling and
  Alarm-Triggered Responses Under Spatial and Detection Uncertainties}. In
  \bibinfo{booktitle}{\emph{Proceedings of AAAI 2016}}.
  \bibinfo{pages}{404--410}.
\newblock


\bibitem[\protect\citeauthoryear{Biswas, Aggarwal, Varakantham, and
  Tambe}{Biswas et~al\mbox{.}}{2021}]%
        {BAVT:learn-preventive-healthcare}
\bibfield{author}{\bibinfo{person}{A. Biswas}, \bibinfo{person}{G. Aggarwal},
  \bibinfo{person}{P. Varakantham}, {and} \bibinfo{person}{M. Tambe}.}
  \bibinfo{year}{2021}\natexlab{}.
\newblock \showarticletitle{Learn to Intervene: An Adaptive Learning Policy for
  Restless Bandits in Application to Preventive Healthcare}. In
  \bibinfo{booktitle}{\emph{Proceedings of the International Joint Conference
  on Artificial Intelligence (IJCAI 2021)}}.
\newblock


\bibitem[\protect\citeauthoryear{Bosansky, Lisy, Jakob, and Pechoucek}{Bosansky
  et~al\mbox{.}}{2011}]%
        {bosansky2011aamas}
\bibfield{author}{\bibinfo{person}{B. Bosansky}, \bibinfo{person}{V. Lisy},
  \bibinfo{person}{M. Jakob}, {and} \bibinfo{person}{M. Pechoucek}.}
  \bibinfo{year}{2011}\natexlab{}.
\newblock \showarticletitle{{Computing Time-Dependent Policies for Patrolling
  Games with Mobile Targets}}. In \bibinfo{booktitle}{\emph{Proceedings of
  AAMAS 2011}}.
\newblock


\bibitem[\protect\citeauthoryear{Br{\'{a}}zdil, Ku{\v{c}}era, and
  {\v{R}}eh{\'{a}}k}{Br{\'{a}}zdil et~al\mbox{.}}{2018}]%
        {BKR:patrol-Internet}
\bibfield{author}{\bibinfo{person}{T. Br{\'{a}}zdil}, \bibinfo{person}{A.
  Ku{\v{c}}era}, {and} \bibinfo{person}{V. {\v{R}}eh{\'{a}}k}.}
  \bibinfo{year}{2018}\natexlab{}.
\newblock \showarticletitle{Solving Patrolling Problems in the Internet
  Environment}. In \bibinfo{booktitle}{\emph{Proceedings of the International
  Joint Conference on Artificial Intelligence (IJCAI 2018)}}.
  \bibinfo{pages}{121--127}.
\newblock


\bibitem[\protect\citeauthoryear{Chen, Cheng, and Wise}{Chen
  et~al\mbox{.}}{2017}]%
        {CCW:cooperative-police-patrol}
\bibfield{author}{\bibinfo{person}{H. Chen}, \bibinfo{person}{T. Cheng}, {and}
  \bibinfo{person}{S. Wise}.} \bibinfo{year}{2017}\natexlab{}.
\newblock \showarticletitle{Developing an Online Cooperative Police Patrol
  Routing Strategy}.
\newblock \bibinfo{journal}{\emph{Computers, Environment and Urban Systems}}
  \bibinfo{volume}{62} (\bibinfo{year}{2017}), \bibinfo{pages}{19--29}.
\newblock


\bibitem[\protect\citeauthoryear{Fang, Jiang, and Tambe}{Fang
  et~al\mbox{.}}{2013}]%
        {Fang2013}
\bibfield{author}{\bibinfo{person}{Fei Fang}, \bibinfo{person}{Albert~Xin
  Jiang}, {and} \bibinfo{person}{Milind Tambe}.}
  \bibinfo{year}{2013}\natexlab{}.
\newblock \showarticletitle{{Optimal Patrol Strategy for Protecting Moving
  Targets with Multiple Mobile Resources}}. In
  \bibinfo{booktitle}{\emph{Proceedings of AAMAS 2013}}.
\newblock


\bibitem[\protect\citeauthoryear{Fave, Jiang, Yin, Zhang, Tambe, Kraus, and
  Sullivan}{Fave et~al\mbox{.}}{2014}]%
        {DJYZTKS:patrolling-uncertainty-JAIR}
\bibfield{author}{\bibinfo{person}{F.M.{}~Delle Fave}, \bibinfo{person}{A.X.{}
  Jiang}, \bibinfo{person}{Z. Yin}, \bibinfo{person}{C~. Zhang},
  \bibinfo{person}{M. Tambe}, \bibinfo{person}{S. Kraus}, {and}
  \bibinfo{person}{J. Sullivan}.} \bibinfo{year}{2014}\natexlab{}.
\newblock \showarticletitle{Game-Theoretic Security Patrolling with Dynamic
  Execution Uncertainty and a Case Study on a Real Transit System}.
\newblock \bibinfo{journal}{\emph{Journal of Artificial Intelligence Research}}
   \bibinfo{volume}{50} (\bibinfo{year}{2014}), \bibinfo{pages}{321--367}.
\newblock


\bibitem[\protect\citeauthoryear{Ford, Kar, Fave, Yang, and Tambe}{Ford
  et~al\mbox{.}}{2014}]%
        {FKDYT:PAWS}
\bibfield{author}{\bibinfo{person}{B. Ford}, \bibinfo{person}{D. Kar},
  \bibinfo{person}{F.M.{}~Delle Fave}, \bibinfo{person}{R. Yang}, {and}
  \bibinfo{person}{M. Tambe}.} \bibinfo{year}{2014}\natexlab{}.
\newblock \showarticletitle{{PAWS:} Adaptive Game-Theoretic Patrolling for
  Wildlife Protection}. In \bibinfo{booktitle}{\emph{Proceedings of AAMAS
  2014}}. \bibinfo{pages}{1641--1642}.
\newblock


\bibitem[\protect\citeauthoryear{Ho and Ouaknine}{Ho and Ouaknine}{2015}]%
        {HO:UAV-problem-PSPACE}
\bibfield{author}{\bibinfo{person}{Hsi-Ming Ho} {and} \bibinfo{person}{J.
  Ouaknine}.} \bibinfo{year}{2015}\natexlab{}.
\newblock \showarticletitle{The Cyclic-Routing {UAV} Problem is
  {PSPACE}-Complete}. In \bibinfo{booktitle}{\emph{Proceedings of {FoSSaCS
  2015}}} \emph{(\bibinfo{series}{Lecture Notes in Computer Science},
  Vol.~\bibinfo{volume}{9034})}. \bibinfo{publisher}{Springer},
  \bibinfo{pages}{328--342}.
\newblock


\bibitem[\protect\citeauthoryear{Huang, Zhou, Hao, and Hou}{Huang
  et~al\mbox{.}}{2019}]%
        {HZHH:multi-robot-patrol-survey}
\bibfield{author}{\bibinfo{person}{L. Huang}, \bibinfo{person}{M. Zhou},
  \bibinfo{person}{K. Hao}, {and} \bibinfo{person}{E. Hou}.}
  \bibinfo{year}{2019}\natexlab{}.
\newblock \showarticletitle{A Survey of Multi-robot Regular and Adversarial
  Patrolling}.
\newblock \bibinfo{journal}{\emph{IEEE/CAA Journal of Automatica Sinica}}
  \bibinfo{volume}{6}, \bibinfo{number}{4} (\bibinfo{year}{2019}),
  \bibinfo{pages}{894--903}.
\newblock


\bibitem[\protect\citeauthoryear{Jakob, Vanek, and Pechoucek}{Jakob
  et~al\mbox{.}}{2011}]%
        {JVP:maritime-security-IS}
\bibfield{author}{\bibinfo{person}{M. Jakob}, \bibinfo{person}{O. Vanek}, {and}
  \bibinfo{person}{M. Pechoucek}.} \bibinfo{year}{2011}\natexlab{}.
\newblock \showarticletitle{Using Agents to Improve International
  MaritimeTransport Security}.
\newblock \bibinfo{journal}{\emph{IEEE Intelligent Systems}}
  \bibinfo{volume}{26}, \bibinfo{number}{1} (\bibinfo{year}{2011}),
  \bibinfo{pages}{90--96}.
\newblock


\bibitem[\protect\citeauthoryear{Karwowski, Mandziuk, Zychowski, Grajek, and
  An}{Karwowski et~al\mbox{.}}{2019}]%
        {KMZGA:moving_targets}
\bibfield{author}{\bibinfo{person}{J. Karwowski}, \bibinfo{person}{J.
  Mandziuk}, \bibinfo{person}{A. Zychowski}, \bibinfo{person}{F. Grajek}, {and}
  \bibinfo{person}{B. An}.} \bibinfo{year}{2019}\natexlab{}.
\newblock \showarticletitle{A Memetic Approach for Sequential Security Games on
  a Plane with Moving Targets}. In \bibinfo{booktitle}{\emph{Proceedings of
  AAAI 2019}}. \bibinfo{pages}{970--977}.
\newblock


\bibitem[\protect\citeauthoryear{Kingma and Ba}{Kingma and Ba}{2015}]%
        {Adam}
\bibfield{author}{\bibinfo{person}{Diederik~P. Kingma} {and}
  \bibinfo{person}{Jimmy Ba}.} \bibinfo{year}{2015}\natexlab{}.
\newblock \showarticletitle{Adam: A Method for Stochastic Optimization}. In
  \bibinfo{booktitle}{\emph{Proceedings of {ICLR} 2015}}.
\newblock


\bibitem[\protect\citeauthoryear{Kla{\v{s}}ka, Ku{\v{c}}era, Lamser, and
  {\v{R}}eh{\'{a}}k}{Kla{\v{s}}ka et~al\mbox{.}}{2018}]%
        {KKLR:patrol-gradient}
\bibfield{author}{\bibinfo{person}{D. Kla{\v{s}}ka}, \bibinfo{person}{A.
  Ku{\v{c}}era}, \bibinfo{person}{T. Lamser}, {and} \bibinfo{person}{V.
  {\v{R}}eh{\'{a}}k}.} \bibinfo{year}{2018}\natexlab{}.
\newblock \showarticletitle{Automatic Synthesis of Efficient Regular Strategies
  in Adversarial Patrolling Games}. In \bibinfo{booktitle}{\emph{Proceedings of
  AAMAS 2018}}. \bibinfo{pages}{659--666}.
\newblock


\bibitem[\protect\citeauthoryear{Kla{\v{s}}ka, Ku{\v{c}}era, Musil, and
  {\v{R}}eh{\'{a}}k}{Kla{\v{s}}ka et~al\mbox{.}}{2021}]%
        {KKMR:Regstar-UAI}
\bibfield{author}{\bibinfo{person}{D. Kla{\v{s}}ka}, \bibinfo{person}{A.
  Ku{\v{c}}era}, \bibinfo{person}{V. Musil}, {and} \bibinfo{person}{V.
  {\v{R}}eh{\'{a}}k}.} \bibinfo{year}{2021}\natexlab{}.
\newblock \showarticletitle{Regstar: Efficient Strategy Synthesis for
  Adversarial Patrolling Games}. In \bibinfo{booktitle}{\emph{Proceedings of
  UAI 2021}}.
\newblock


\bibitem[\protect\citeauthoryear{Kla{\v{s}}ka, Ku{\v{c}}era, and
  {\v{R}}eh{\'{a}}k}{Kla{\v{s}}ka et~al\mbox{.}}{2020}]%
        {KKR:patrol-drones}
\bibfield{author}{\bibinfo{person}{D. Kla{\v{s}}ka}, \bibinfo{person}{A.
  Ku{\v{c}}era}, {and} \bibinfo{person}{V. {\v{R}}eh{\'{a}}k}.}
  \bibinfo{year}{2020}\natexlab{}.
\newblock \showarticletitle{Adversarial Patrolling with Drones}. In
  \bibinfo{booktitle}{\emph{Proceedings of AAMAS 2020}}.
  \bibinfo{pages}{629--637}.
\newblock


\bibitem[\protect\citeauthoryear{Ku{\v{c}}era and Lamser}{Ku{\v{c}}era and
  Lamser}{2016}]%
        {KL:patrol-regular}
\bibfield{author}{\bibinfo{person}{A. Ku{\v{c}}era} {and} \bibinfo{person}{T.
  Lamser}.} \bibinfo{year}{2016}\natexlab{}.
\newblock \showarticletitle{Regular Strategies and Strategy Improvement:
  Efficient Tools for Solving Large Patrolling Problems}. In
  \bibinfo{booktitle}{\emph{Proceedings of AAMAS 2016}}.
  \bibinfo{pages}{1171--1179}.
\newblock


\bibitem[\protect\citeauthoryear{Maza, Caballero, Capit{\'{a}}n, de~Dios, and
  Ollero}{Maza et~al\mbox{.}}{2011}]%
        {MCCMO:multiUAV-disaster}
\bibfield{author}{\bibinfo{person}{I. Maza}, \bibinfo{person}{F. Caballero},
  \bibinfo{person}{J. Capit{\'{a}}n}, \bibinfo{person}{J.R.{}~Mart{\'{\i}}nez
  de Dios}, {and} \bibinfo{person}{A. Ollero}.}
  \bibinfo{year}{2011}\natexlab{}.
\newblock \showarticletitle{Experimental Results in Multi-{UAV} Coordination
  for Disaster Management and Civil Security Applications}.
\newblock \bibinfo{journal}{\emph{Journal of Intelligent and Robotic Systems}}
  \bibinfo{volume}{61}, \bibinfo{number}{1--4} (\bibinfo{year}{2011}),
  \bibinfo{pages}{563--585}.
\newblock


\bibitem[\protect\citeauthoryear{Munoz~de Cote, Stranders, Basilico, Gatti, and
  Jennings}{Munoz~de Cote et~al\mbox{.}}{2013}]%
        {MunozdeCote2013}
\bibfield{author}{\bibinfo{person}{Enrique Munoz~de Cote},
  \bibinfo{person}{Ruben Stranders}, \bibinfo{person}{Nicola Basilico},
  \bibinfo{person}{Nicola Gatti}, {and} \bibinfo{person}{Nick Jennings}.}
  \bibinfo{year}{2013}\natexlab{}.
\newblock \showarticletitle{Introducing alarms in adversarial patrolling games:
  extended abstract}. In \bibinfo{booktitle}{\emph{Proceedings of AAMAS 2013}}.
  \bibinfo{pages}{1275--1276}.
\newblock


\bibitem[\protect\citeauthoryear{Norris}{Norris}{1998}]%
        {Norris:Book}
\bibfield{author}{\bibinfo{person}{J.R.{} Norris}.}
  \bibinfo{year}{1998}\natexlab{}.
\newblock \bibinfo{booktitle}{\emph{{Markov} Chains}}.
\newblock \bibinfo{publisher}{Cambridge University Press}.
\newblock


\bibitem[\protect\citeauthoryear{Pita, Jain, Marecki, Ord{\'{o}}nez, Portway,
  Tambe, Western, Paruchuri, and Kraus}{Pita et~al\mbox{.}}{2008}]%
        {PJMOPTWPK:Deployed-ARMOR}
\bibfield{author}{\bibinfo{person}{J. Pita}, \bibinfo{person}{M. Jain},
  \bibinfo{person}{J. Marecki}, \bibinfo{person}{F. Ord{\'{o}}nez},
  \bibinfo{person}{C. Portway}, \bibinfo{person}{M. Tambe}, \bibinfo{person}{C.
  Western}, \bibinfo{person}{P. Paruchuri}, {and} \bibinfo{person}{S. Kraus}.}
  \bibinfo{year}{2008}\natexlab{}.
\newblock \showarticletitle{Deployed {ARMOR} Protection: The Application of a
  Game Theoretic Model for Security at the {Los} {Angeles} {Int.} {Airport}}.
  In \bibinfo{booktitle}{\emph{Proceedings of AAMAS 2008}}.
  \bibinfo{pages}{125--132}.
\newblock


\bibitem[\protect\citeauthoryear{Portugal and Rocha}{Portugal and
  Rocha}{2011}]%
        {PR:multi-patrolling-survey}
\bibfield{author}{\bibinfo{person}{D. Portugal} {and} \bibinfo{person}{R.
  Rocha}.} \bibinfo{year}{2011}\natexlab{}.
\newblock \showarticletitle{A Survey on Multi-Robot Patrolling Algorithms}.
\newblock \bibinfo{journal}{\emph{Technological Innovation for Sustainability}}
   \bibinfo{volume}{349} (\bibinfo{year}{2011}), \bibinfo{pages}{139--146}.
\newblock


\bibitem[\protect\citeauthoryear{Sinha, Fang, An, Kiekintveld, and Tambe}{Sinha
  et~al\mbox{.}}{2018}]%
        {SFAKT:Stackelberg-Security-Games}
\bibfield{author}{\bibinfo{person}{A. Sinha}, \bibinfo{person}{F. Fang},
  \bibinfo{person}{B. An}, \bibinfo{person}{C. Kiekintveld}, {and}
  \bibinfo{person}{M. Tambe}.} \bibinfo{year}{2018}\natexlab{}.
\newblock \showarticletitle{Stackelberg Security Games: Looking Beyond a Decade
  of Success}. In \bibinfo{booktitle}{\emph{Proceedings of the International
  Joint Conference on Artificial Intelligence (IJCAI 2018)}}.
  \bibinfo{pages}{5494--5501}.
\newblock


\bibitem[\protect\citeauthoryear{Tambe}{Tambe}{2011}]%
        {Tambe:book}
\bibfield{author}{\bibinfo{person}{M. Tambe}.} \bibinfo{year}{2011}\natexlab{}.
\newblock \bibinfo{booktitle}{\emph{Security and Game Theory. Algorithms,
  Deployed Systems, Lessons Learned}}.
\newblock \bibinfo{publisher}{Cambridge University Press}.
\newblock


\bibitem[\protect\citeauthoryear{Tsai, Rathi, Kiekintveld, Ord{\'o}{\~{n}}ez,
  and Tambe}{Tsai et~al\mbox{.}}{2009}]%
        {TRKOT:IRIS}
\bibfield{author}{\bibinfo{person}{J. Tsai}, \bibinfo{person}{S. Rathi},
  \bibinfo{person}{C. Kiekintveld}, \bibinfo{person}{F. Ord{\'o}{\~{n}}ez},
  {and} \bibinfo{person}{M. Tambe}.} \bibinfo{year}{2009}\natexlab{}.
\newblock \showarticletitle{{IRIS}---A Tool for Strategic Security Allocation
  in Transportation Networks Categories and Subject Descriptors}. In
  \bibinfo{booktitle}{\emph{Proceedings of AAMAS 2009}}.
  \bibinfo{pages}{37--44}.
\newblock


\bibitem[\protect\citeauthoryear{Wang, Shi, Yu, Wu, Singh, Joppa, and
  Fang}{Wang et~al\mbox{.}}{2019}]%
        {WSYWSJF:Patrolling-learning}
\bibfield{author}{\bibinfo{person}{Y. Wang}, \bibinfo{person}{Z.R.{} Shi},
  \bibinfo{person}{L. Yu}, \bibinfo{person}{Y. Wu}, \bibinfo{person}{R. Singh},
  \bibinfo{person}{L. Joppa}, {and} \bibinfo{person}{F. Fang}.}
  \bibinfo{year}{2019}\natexlab{}.
\newblock \showarticletitle{Deep Reinforcement Learning for Green Security
  Games with Real-Time Information}. In \bibinfo{booktitle}{\emph{Proceedings
  of AAAI 2019}}. \bibinfo{pages}{1401--1408}.
\newblock


\bibitem[\protect\citeauthoryear{Xu}{Xu}{2021}]%
        {Xu:Green-security}
\bibfield{author}{\bibinfo{person}{L. Xu}.} \bibinfo{year}{2021}\natexlab{}.
\newblock \showarticletitle{Learning and Planning Under Uncertainty for Green
  Security}. In \bibinfo{booktitle}{\emph{Proceedings of the International
  Joint Conference on Artificial Intelligence (IJCAI 2021)}}.
\newblock


\bibitem[\protect\citeauthoryear{Yan, Jouandeau, and Cherif}{Yan
  et~al\mbox{.}}{2013}]%
        {YJC:multi-robot-coordination-survey}
\bibfield{author}{\bibinfo{person}{Z. Yan}, \bibinfo{person}{N. Jouandeau},
  {and} \bibinfo{person}{A.A.{} Cherif}.} \bibinfo{year}{2013}\natexlab{}.
\newblock \showarticletitle{A Survey and Analysis of Multi-Robot Coordination}.
\newblock \bibinfo{journal}{\emph{International Journal of Advanced Robotic
  Systems}} \bibinfo{volume}{10}, \bibinfo{number}{12} (\bibinfo{year}{2013}),
  \bibinfo{pages}{1--18}.
\newblock


\bibitem[\protect\citeauthoryear{Yin, Korzhyk, Kiekintveld, Conitzer, and
  Tambe}{Yin et~al\mbox{.}}{2010}]%
        {YKKCT:Stackelberg-Nash-security}
\bibfield{author}{\bibinfo{person}{Z. Yin}, \bibinfo{person}{D. Korzhyk},
  \bibinfo{person}{C. Kiekintveld}, \bibinfo{person}{V. Conitzer}, {and}
  \bibinfo{person}{M. Tambe}.} \bibinfo{year}{2010}\natexlab{}.
\newblock \showarticletitle{Stackelberg vs.{} {Nash} in security games:
  Interchangeability, equivalence, and uniqueness}. In
  \bibinfo{booktitle}{\emph{Proceedings of AAMAS 2010}}.
  \bibinfo{pages}{1139--1146}.
\newblock


\end{thebibliography}
\end{document}